\documentclass{IEEEtran}
\usepackage{amsmath,amsfonts}
\usepackage{amssymb}
\usepackage{cite}
\usepackage{algorithmicx}
\usepackage{algorithm}
\usepackage{algpseudocode}
\usepackage{graphicx}
\usepackage{subfigure}
\usepackage{makecell}
\usepackage{multirow}
\usepackage{multicol}
\usepackage{cases}
\usepackage{stfloats}
\usepackage{enumerate}
\usepackage{url}
\usepackage{booktabs}
\usepackage{diagbox}
\usepackage{bbding}
\usepackage{color}
\usepackage{xcolor}
\usepackage{framed}
\usepackage{amsthm}
\usepackage[normalem]{ulem}

\newtheorem{theorem}{Theorem}
\newtheorem{lemma}{Lemma}
\allowdisplaybreaks[4]

\makeatletter
\renewcommand{\maketag@@@}[1]{\hbox{\m@th\normalsize\normalfont#1}}
\makeatother



\begin{document}

\title{

Non-Orthogonal Multiple Access-Based  Continuous-Variable Quantum Key Distribution: 
Secret Key Rate Analysis and Power Allocation 
}

\author{Zhichao Dong, Xiaolin Zhou, \IEEEmembership{Senior Member, IEEE,} Huang Peng,
Wei Ni, \IEEEmembership{Fellow, IEEE,} \\
Ekram Hossain, \IEEEmembership{Fellow, IEEE,} 
and 
Xin Wang, \IEEEmembership{Fellow, IEEE} 
\thanks{Z. Dong, X. Zhou, and X. Wang are with the Key Laboratory for Information Science of Electromagnetic Waves, College of  Future Information Technology, Fudan University, Shanghai 200433, China (e-mail: zhichaodong@ynu.edu.cn; zhouxiaolin@fudan.edu.cn; xwang11@fudan.edu.cn)}
\thanks{H. Peng is with the State Key Laboratory of Photonics and Communications, Institute for Quantum Sensing and Information Processing, Shanghai Jiao Tong University, Shanghai 200240, China, and the 
Shanghai Research Center for Quantum Sciences, Shanghai 201315, China
Hefei National Laboratory, Hefei 230088, China (huang.peng@sjtu.edu.cn)}
\thanks{W. Ni is with Data61, Commonwealth Scientific and Industrial Research Organisation, Sydney, NSW 2122, Australia, and the School of Computer Science and Engineering, University of New South Wales, Sydney, NSW 2052, Australia (e-mail: wei.ni@ieee.org).}
\thanks{E. Hossain is with the Department of Electrical and Computer Engineering, University of Manitoba, Winnipeg, MB R3T 2N2, Canada (e-mail:
ekram.hossain@umanitoba.ca).}}


\maketitle

\begin{abstract}
Multi-user quantum key distribution (QKD) under malicious quantum attack poses a critical challenge for the large-scale quantum Internet. This paper maximizes the sum secret key rate (SKR) for a novel uplink non-orthogonal multiple access-based continuous-variable QKD (NOMA-CVQKD) system under collective attacks. The new system employs Gaussian-modulated coherent states and a quantum successive interference cancellation-based heterodyne receiver. Closed-form asymptotic bounds are derived for the legitimate users’ achievable key rates via the entropy power inequality and maximum entropy principle, and for the eavesdropper’s intercepted information based on Holevo information. A successive convex approximation-based power allocation algorithm is developed to maximize the asymptotic sum SKR of the NOMA-CVQKD system under collective attacks, with guaranteed convergence to a locally optimal Karush–Kuhn–Tucker solution. Simulations show that the new NOMA-CVQKD system with the power allocation algorithm achieves up to a 23\% higher sum SKR than quantum-orthogonal multiple access, supports 16 users at excess noise variance $W = 0.1$, and maintains robustness under varying turbulence intensities and transmission distances.

\end{abstract}

\begin{IEEEkeywords}
Non-orthogonal multiple access (NOMA), continuous variable quantum key distribution (CVQKD), collective attack, secret key rate (SKR).
\end{IEEEkeywords}

\section{Introduction}

\IEEEPARstart{W}{ith} the rapid global development of the quantum Internet, the number of quantum terminals is increasing dramatically~\cite{ref_1,ref_2}. To accommodate massive access with high spectral efficiency, conventional quantum key distribution (QKD) schemes based on orthogonal multiple access (OMA), e.g., quantum time-division multiple access (TDMA) and quantum wavelength-division multiple access (WDMA)~\cite{TDMA_QKD,WDMA_QKD2}, are facing inherent technical bottlenecks. Meanwhile, malicious quantum attacks pose a serious threat to the reliability of QKD in large-scale quantum networks~\cite{Hanzo_survey2,Survey_Zhou}. These challenges have motivated growing interest in quantum non-orthogonal multiple access (NOMA) for key distribution, under adversarial quantum attacks.

Quantum NOMA can offer flexible resource allocation, higher spectral efficiency, larger user connectivity, and reduced access delay, compared with quantum OMA~\cite{Multiaccess,Survey}. In NOMA-based continuous-variable QKD (CVQKD) systems,  multiple QKD transmitters can share the same time–frequency resource and dynamically allocate transmit power according to their respective channel conditions (e.g., path loss and random fading)~\cite{Survey_Zhou}. Superposition coding is employed at the transmitters, while the receiver applies successive interference cancellation (SIC) to decode the signals~\cite{Hanzo_survey2,NOMA_SIC}. To date, most QKD multiple access schemes remain based on quantum OMA, and no research has evaluated the secret key rate (SKR) in NOMA-CVQKD systems. 

Although CVQKD protocols offer theoretical unconditional security under ideal physical device implementations and lossless quantum channels~\cite{ref_1,ref_2}, practical free-space CVQKD systems are susceptible to imperfections such as beam spreading, pointing errors, and device defects, which can be exploited by an eavesdropper (Eve)\cite{Collective_att4_opt}. Collective attack is a powerful practical adversary model due to its strong concealment characteristics~\cite{MIMO2}. In a collective attack, Eve establishes quantum correlations between her ancilla states and the transmitted quantum states via a unitary operation, enabling her to extract information while preserving the statistical properties of the original signals, thus avoiding detection~\cite{Terahertz-QKD,Collective_att1,Collective_att2}. 
It is of practical interest to analyze the SKR of NOMA-CVQKD systems under collective attacks.

\subsection{Related Work} 
\subsubsection{Quantum OMA}
Quantum OMA is a widely adopted approach in quantum communications, where orthogonal resources (e.g., time, frequency, or wavelength slots) are allocated to users to avoid interference.
Concha and Poor~\cite{Multiaccess} proposed a quantum multiple access channel model that accounts for multiple access interference, signal attenuation, and random noise, providing a foundation for the analysis and design of quantum multi-user communication systems. To date, most studies on CVQKD multi-access technologies have focused on quantum OMA. For example, Eriksson et al.\cite{WDMA_QKD} demonstrated the simultaneous transmission of CVQKD over multiple parallel wavelength channels. Bathaee et al.\cite{WDMA_QKD2} analyzed a quantum WDMA system and highlighted its scalability potential through wavelength allocation. Razavi et al.~\cite{TDMA_QKD} proposed a generic TDMA framework for QKD, reporting a linear decrease in key rate with an increasing number of users. Extending this concept, Bahrani et al.~\cite{OFDMA_QKD} developed an orthogonal frequency-division multiple access (OFDMA)-QKD system applicable to CVQKD.

\subsubsection{Quantum NOMA}
Via power-domain multiplexing, NOMA enables simultaneous resource reuse among users for improved spectral efficiency and scalability. Originally designed for radio communications, NOMA has been increasingly considered for quantum communications.
Shen et al.\cite{QNOMA1} proposed an interleave-division multiple access (IDMA)-based iterative quantum NOMA scheme capable of mitigating quantum interference in coherent-state multiple access communications. This approach demonstrates high robustness against non-ideal quantum detection impairments, such as dark counts and quantum mode mismatch. Yu et al.\cite{QNOMA2} developed a NOMA-QKD system for atmospheric channels, introducing an iterative parallel interference cancellation (PIC) key detection algorithm with rapid convergence, which achieves a higher SKR than conventional WDMA-QKD methods. However, compared with the substantial progress achieved in quantum OMA, the theoretical foundations and practical methods for efficient quantum NOMA remain underdeveloped.

\subsubsection{Collective Attack}
A collective attack constitutes a critical threat to CVQKD processes, where Eve employs quantum memory to store intercepted states and subsequently performs optimal joint measurements to maximize information extraction while remaining undetectable~\cite{Collective_att3,Collective_att4_opt}. It was demonstrated in~\cite{Collective_att4_opt} that collective attacks represent an asymptotically optimal eavesdropping strategy. Ottaviani et al.~\cite{Terahertz-QKD} characterized collective attacks in free-space turbulent channels, while Kundu et al.~\cite{MIMO1,MIMO2} extended the model to multiple-input multiple-output (MIMO) channels, and quantified the SKR of CVQKD systems under such attacks.

To the best of our knowledge, existing studies, e.g.,~\cite{Pro_shaping,Terahertz-QKD,Multi-Ring,MIMO1,MIMO2,Collective_att4_opt,Collective_att1,Collective_att2,Collective_att3}, have not 
considered NOMA-CVQKD, particularly 
under collective attacks. 
This paper addresses these gaps, for the first time, by proposing a framework for secure and efficient NOMA-CVQKD operation under collective attacks.  

\subsection{Contributions}
This paper presents a novel uplink NOMA-CVQKD transmission framework that can operate under collective attacks, atmospheric turbulence, and excess noise impairments. Leveraging an SIC-based QKD scheme, we formulate a new power allocation problem and develop an effective algorithm to maximize the sum SKR of the system. The key contributions of the paper are as follows:

\begin{itemize}
\item We design a new uplink NOMA-CVQKD scheme employing Gaussian-modulated coherent states for transmissions, and introduce an SIC-based heterodyne receiver architecture tailored for multi-user CVQKD systems.
\item We analyze the sum SKR of the uplink NOMA-CVQKD system. We also establish an asymptotic lower bound on legitimate users’ achievable key rates via the entropy power inequality and maximum entropy principle, and an asymptotic upper bound on Eve’s accessible information based on Holevo information.
\item Given the bounds, we obtain the asymptotic sum SKR of the uplink NOMA-CVQKD system. A successive convex approximation (SCA)-based power allocation algorithm is developed to maximize the asymptotic sum SKR.
\item We confirm analytically the convergence of the proposed power allocation algorithm to a locally optimal solution that satisfies the Karush–Kuhn–Tucker (KKT) conditions.
\end{itemize}

Extensive simulations corroborate that the new uplink NOMA-CVQKD scheme with the proposed locally optimal power allocation algorithm achieves an improvement of up to 23\% in sum SKR, compared to the benchmarks, and supports up to 16 users under excess noise variance $W = 0.1$. The scheme also maintains robust performance across varying turbulence intensities and transmission distances.

The rest of this paper is organized as follows. Section~\ref{sec:SysModel} describes the model of the uplink NOMA-CVQKD system. In Section~\ref{sec:Analysis}, we derive the lower bound for the users’ key rates and the information bound for Eve, and establish the sum SKR of the system and its asymptotic approximation. In Section~\ref{sec:Algorithm}, we maximize the asymptotic sum SKR and develop the new power allocation algorithm.
In Section~\ref{sec:Simulation}, extensive simulations are presented to verify the effectiveness of the proposed scheme. Section~\ref{sec:Conclusion} concludes this paper.

\begin{table}[t]
  \centering
  \renewcommand{\arraystretch}{1.2}
  \caption{Notation and Definitions}
  \label{tab:notations}
  \begin{tabular}{cp{0.7\linewidth}}
    \hline
    \textbf{Notation} & \textbf{Description} \\
    \hline
    $X_k$, $V_{a}^{\left( k \right)}$ & The Gaussian modulated quadrature of the $k$-th user and its variance\\
    $Y_k$ & The received signal of the $k$-th user\\
    $\boldsymbol{X}$ & The collective quadrature vector \\
$\eta$  & Reconciliation efficiency \\
$n_{\det}$, $\delta _{\det}^2$  & The noise and its variance at the heterodyne receiver  \\
$n_k$, $\delta _k^2$  & The interference of the $k$-th user and its variance\\
$T_k$  &  Channel transmittance \\
$T_{k,t}$  & Transmissivity of the atmospheric turbulence \\
$d_k$  & The link distance\\
$T_{k,l}$  &  Path loss \\
$\hat{\rho}_{X_kY_k}$ & The density matrix of the joint state between the $k$-th user's transmitted coherent state and the BS \\
$\hat{\rho}_{E_k}$ & The density matrix of Eve's  subsystem state\\
$\hat{\rho}_{E'_k}$ & The density matrix of Eve's  ancilla state\\
$V_{\max}^{\left( k \right)}$  & The maximum coherent-state variance of user~$k$\\
$V_{\max}^{\left( \mathrm{BS} \right)}$  & The maximum effective variance at the BS \\
$V_{\mathrm{attack}}$  & The additional noise variance introduced by attackers  \\
$P_{\max}^{\left( \mathrm{BS} \right)}$  & The maximum received power at the BS \\
$\varepsilon_k$, $W_k$  & The excess noise and its variance introduced by Eve's collective attack \\
$\sigma _x$  & The intensity of turbulence\\
    \hline
  \end{tabular}
\end{table}

\begin{figure*}[ht]
\centering
\includegraphics[width=0.8\textwidth]{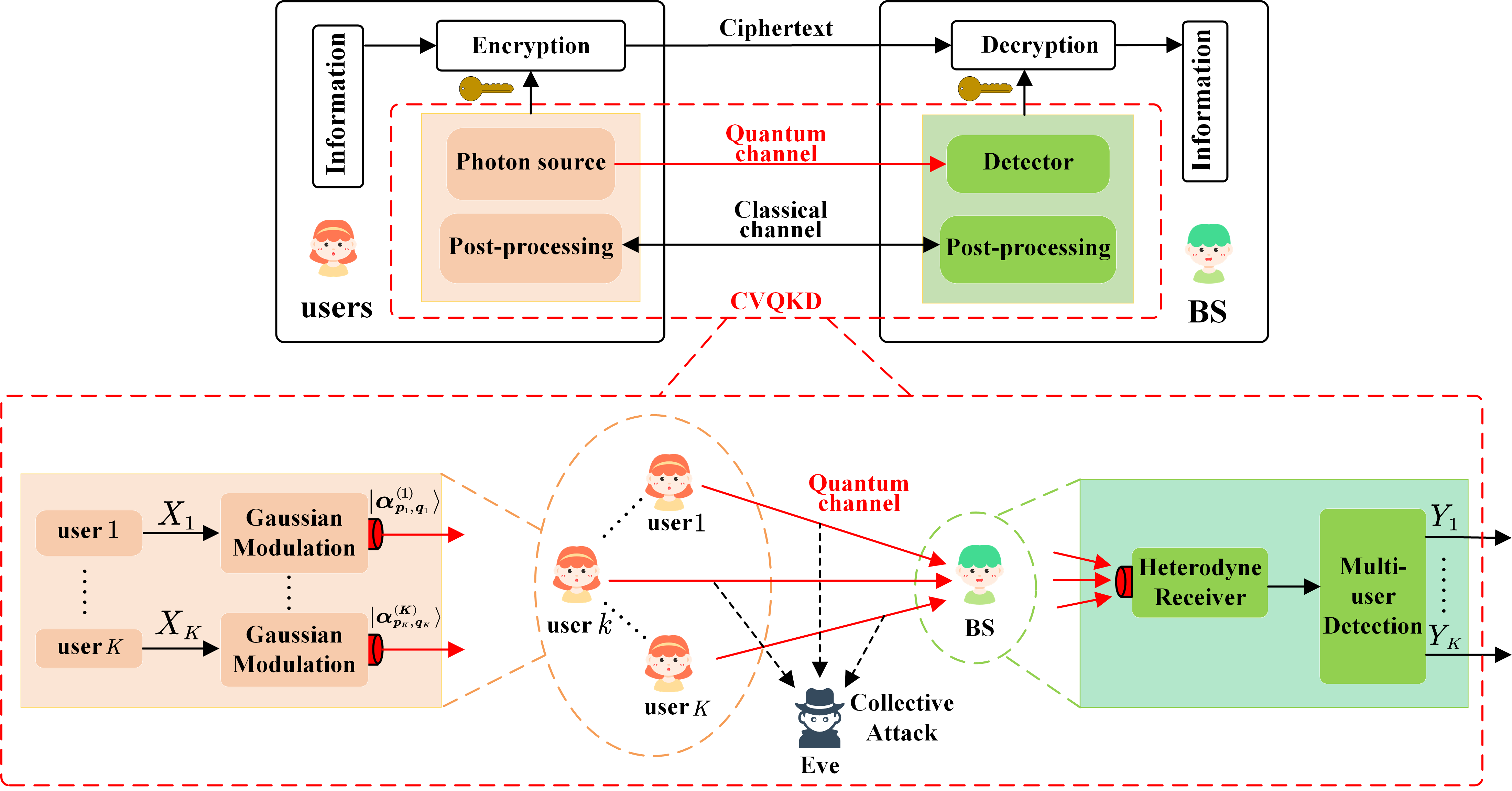}
\caption{An illustration of the considered uplink  NOMA-CVQKD system where $K$ users form an uplink NOMA cluster, with each user transmitting key using Gaussian-modulated coherent states. The BS employs an SIC-based heterodyne receiver to sequentially detect the received coherent states, while Eve performs a collective attack to extract the secret key during the CVQKD process.}
\label{fig:system}
\end{figure*}

\section{System Model and Assumptions}
\label{sec:SysModel}

As depicted in Fig.~\ref{fig:system}, we consider an uplink NOMA-CVQKD system, where there are $K$ transmitters/users and a receiver/base station (BS). 
Following the CVQKD protocol, each user $k=1,\cdots,K$ transmits Gaussian-modulated coherent states carrying quantum secret keys to the BS through a quantum channel. Then, the users and the BS establish secret keys via post-processing. These secret keys are used to encrypt classical messages in the classical channel.
There is also an adversary that launches collective attacks in the quantum channel to intercept and store the quantum secret keys for future decryption. 

\subsection{Transmitters/Users}
According to the principles of quantum mechanics, in the studied NOMA-CVQKD system, the transmitted Gaussian-modulated coherent state can be expressed as \cite{Hanzo_survey1}
 \begin{equation}
 \label{eq:coherent_state}
    |\alpha _{p,q}\rangle =|p+iq\rangle , \;p,q\in \mathbb{R},
    \end{equation}
where $p, q \sim \mathcal{N}(0, V_a)$ are independent, real-value Gaussian random variables, representing the amplitude quadrature and phase quadrature of the coherent state, respectively. 

Let $X_k\in\mathbb{R}, \,k=1,\cdots,K$ denote the transmitted quadrature of the coherent state  $|\alpha _{p_k,q_k}^{\left(k\right)}\rangle$ sent by the $k$-th user with modulation variance  $V_{a}^{\left(k \right)}$. 
Also, define \( \boldsymbol{X} = [X_1, \dots, X_K]^T \) as the collective quadrature vector.

\subsection{Receiver/BS}
The coherent states, subject to atmospheric turbulence (and collective attacks), are detected at the receiver using a heterodyne receiver.
Consider NOMA for the simultaneous transmissions of the $K$ users' coherent states in the single-mode Hilbert space.
The coherent state at the heterodyne receiver is given as $\left|\sum_{k=1}^K{\sqrt{T_k}\alpha _{p_k,q_k}^{\left(k\right)}}\right\rangle $~\cite{Multiaccess,Diversity},  where $T_k=T_{k,t}T_{k,l}$ is the channel transmittance (i.e., a measure of how much of the quantum signal sent from the transmitter makes it to the receiver without being lost) with the transmissivity $T_{k,t}$ of the atmospheric turbulence and the path loss~$T_{k,l}$.

SIC is carried out at the heterodyne receiver, which first ranks the users based on their channel gains to determine the decoding order \cite{NOMA_SIC}. 
Without loss of generality, we suppose that the decoding order is consistent with the indices of the users.
Since each user independently prepares Gaussian-modulated coherent states, the positive-operator valued measure (POVM)-based quantum measurements at the heterodyne receiver can sequentially extract user-specific signals~\cite{heterodyne_receive}.

Following the decoding order, upon decoding the \(k\)-th user's signal,  the signals of the first \((k-1)\) users are decoded and canceled, while those of the remaining \((K-k)\) users are treated as interference.
The received signal of the $k$-th user can be written as 
\begin{equation}
\label{eq:Y_k}
Y_k=\sqrt{T_k}X_k+\sqrt{1-T_k}\varepsilon_k+\sum_{i=k+1}^K{\sqrt{T_i}X_i}+n_{\det},
\end{equation}
where 
$\varepsilon_k$ denotes the excess noise introduced by Eve to extract the key signal with $\mathrm{Var}(\varepsilon_k) = W_k$; $n_{\det}$ is the noise of the heterodyne receiver with variance $\delta _{\det}^2$, accounting for quantum noise and additional classical noise~\cite{heterodyne_receive,Quantum_Communication}. 

The term \(\sum_{i=k+1}^K \sqrt{T_i} X_i\) represents the interference from the signals of the \((k+1)\)-th to the \(K\)-th users during the SIC process, treated as additive noise. Let $n_k=\sum_{i=k+1}^K{\sqrt{T_i}X_i}+\sqrt{1-T_k}\varepsilon_k+n_{\det}$  denote the total interference to the signal of the $k$-th user, $\forall k$, with variance \( \delta _k^2 = \mathrm{Var}(n_k) = \sum_{i=k+1}^K T_i V_a^{(i)} + (1-T_k)W_k + \delta_{\det}^2 \).

The probability density function (PDF) of the received signals of the heterodyne receiver   is given by \cite{Quantum_Communication,Pro_shaping}   
    \begin{align}\label{eq:basePr_group}
    &\mathrm{Pr}\big( Y\!=\!y\,| \boldsymbol{X} \big)=\frac{1}{\sqrt{2\pi \delta ^2}}\exp\! \bigg[\! -\frac{\big( y\!-\!\!\sum_{k=1}^K{\!\sqrt{T_k}X_k} \big) ^2}{2\delta ^2} \bigg],
    \end{align}
where $\delta^2 = \delta_{\text{det}}^2 + \sum_{k=1}^K (1-T_k) W_k$ is the noise variance at the heterodyne receiver.

\subsection{Attack Model}
Eve launches a collective attack through the quantum channel during the CVQKD process. 
It interacts with each transmitted coherent state $|\alpha_{p_k,q_k}\rangle,\forall k$ between the users and BS, e.g., by placing a beam splitter in the quantum channel, enabling partial interception of the signals and inevitably introducing excess noise $\varepsilon_k$~\cite{MIMO1}.

\subsubsection{Intercepted Quantum Information} For the $k$-th user, let  $\hat{\rho}_{X_kY_k}$ denote the density matrix of the joint state between the $k$-th user' transmitted coherent state $|\alpha _{p_k,q_k}^{\left(k\right)}\rangle$ and the BS.
Eve prepares a pair of entangled quantum states, keeping
one subsystem as Eve’s ancilla state $\hat{\rho}_{E'k}$ in a quantum memory for subsequent optimal measurement.
Coupling the other subsystem \(\hat{\rho}_{E_k}\) to the joint state of the $k$-th user and the BS, $\hat{\rho}_{X_kY_k}$, through a unitary operation implemented by the beam splitter enables Eve to intercept partial quantum information through subsequent measurements on $\hat{\rho}_{E'k}$ while maintaining the joint system $\hat{\rho}_{X_kY_kE_k}$ in a pure state, thus satisfying \(\text{Tr}(\hat{\rho}_{X_kY_kE_k}^2)=1\)~\cite{MIMO2,Quantum_Communication}. Here, \(\text{Tr}(\cdot)\) denotes trace operator. 

Assume that Eve possesses perfect channel state information~(CSI). To keep the collective attack undetectable, Eve can adjust the beam splitter's transmission coefficient to match the channel transmittance $T_k$, maintain it within normal statistical fluctuations, and constrain the induced excess noise within the expected noise of the legitimate users,
keeping Eve's presence and position undetectable~\cite{Collective_att1,Collective_att2,Collective_att3}.

\subsubsection{Quantum State Storage and Optimal Quantum Measurement} Eve stores the ancilla state $\hat{\rho}_{E'_k}$ in a quantum memory, preserving the quantum correlation with the joint state $\hat{\rho}_{X_kY_k}$. Later, Eve can perform
an optimal joint POVM measurement on the stored composite states to maximize the secret key information \cite{Collective_att4_opt}.

\subsubsection{Reverse Reconciliation} In the considered NOMA-CVQKD system, the BS employs reverse reconciliation (RR) to generate secret keys based on the BS’s measurements~\cite{RR1,RR3_25km}. The RR protocol provides positive secure key distribution under any channel transmittance $0<T_k < 1$\cite{MIMO2}. On the other hand, Eve can perform a joint optimal POVM measurement on the stored state by applying the BS's publicly announced quadrature choices. 

Assume that the collective attack strategy can inherently adapt to the NOMA-CVQKD system, discriminating the coherent states of different transmitters through a similar SIC process to the one conducted at the benign receiver.

\section{Analysis of SKR for Uplink NOMA-CVQKD}
\label{sec:Analysis}
In this section, we analyze the sum SKR of the considered uplink NOMA-CVQKD system. We start by deriving the exact expression for the sum-rate based on the definition of mutual information. Then, we establish a lower bound on the achievable rates of the users and an upper bound on Eve’s accessible information. Eventually, we obtain the sum SKR and its asymptotic approximation under collective attacks.

We can employ Shannon entropy to measure the information of the coherent states in the form of their uncertainty \cite{Survey}.
According to the definition of mutual information, we establish the expression for the sum-rate of the uplink NOMA-CVQKD system, as stated in the following theorem.

\begin{lemma}
The sum-rate of the uplink NOMA-CVQKD system is given by
\begin{align}
&I_{\mathrm{sum}}=-\int_{-\infty}^{\infty}{\int_{-\infty}^{\infty}{\cdots \int_{-\infty}^{\infty}{\mathrm{Pr}\big( \boldsymbol{X} \big) \mathrm{Pr}\big( Y|\boldsymbol{X} \big) dX_1\cdots dX_K}}}\nonumber\\
&\times \log \Bigg[ \int_{-\infty}^{\infty}{\cdots \int_{-\infty}^{\infty}{\mathrm{Pr}\big( \boldsymbol{X} \big) \mathrm{Pr}\big( Y|\boldsymbol{X} \big) dX_1\cdots dX_K}} \Bigg] dy+\nonumber\\
&\int_{\!\!\!-\!\infty}^{\infty}\!\!\!\!\int_{\!\!-\!\infty}^{\infty}\!\!\!\!\!\!\!\cdots \!\!\!\int_{\!\!-\!\infty}^{\infty}\!\!\!\!\!\!\!\Pr( \boldsymbol{X} ) \!\Pr( Y|\boldsymbol{X} )
\log\! \mathrm{Pr}\big( Y|\boldsymbol{X} \big) dX_{\!1}\!\cdots dX_{\!K}dy,\label{eq:SumRate}
\end{align}
where $\mathrm{Pr}\left( Y|\boldsymbol{X} \right)$ is given in \eqref{eq:basePr_group}, and 
\begin{equation}
\begin{aligned}
\label{eq:joint_PDF}
\mathrm{Pr}\big( \boldsymbol{X} \big) =\prod_{k=1}^K{\frac{1}{\sqrt{2\pi {V_a^{(k)}}}}\exp \big( -\frac{X_k^2}{2{V_a^{(k)}}} \big)},
\end{aligned}
\end{equation}
gives the joint PDF of the independent Gaussian-modulated quadratures $\boldsymbol{X} = [X_1, \dots, X_K]^T$, with variance $V_{a}^{(k)}$ per $X_k$, $\forall k$. Moreover, $I_{\text{sum}}$ converges, as \(\|\boldsymbol{X}\|\to\infty\) and \(|y|\to\infty\).
\end{lemma}

\begin{proof}
    See \textbf{Appendix~A}.
\end{proof}

\subsection{Lowe-Bound for the Key Rate Under Collective Attack}
We first establish a lower bound on the achievable key rate of any user $k$ in the NOMA-CVQKD system based on quantum mutual information theory, as follows. 
\begin{lemma}
Based on the entropy power inequality (EPI) and the maximum entropy property, a lower bound for the achievable key rate of the $k$-th user  ($k=1,\ldots,K$) in the NOMA-CVQKD system under the collective attack is given by
\begin{equation}
\label{eq:I_k_lower}
I_{k}^{(\mathrm{low})} = \log \Bigg[ 1+\frac{T_kV_{a}^{(k)}}{\sum_{i=k+1}^K V_{a}^{(i)} + (1-T_k) W_k + \delta_{\det}^2} \Bigg].
\end{equation}
\end{lemma}

\begin{proof}
    See \textbf{Appendix~B}.
\end{proof}

This lower bound for the key rate of the $k$-th user is derived through quantum mutual information under heterodyne detection, where the measurement results obey the Gaussian distribution as given in \eqref{eq:basePr_group}. Note that $I_{k}^{(\mathrm{low})}$ accounts for the quantum interference from the signals of the $(k+1)$-th to the $K$-th users during the SIC process at the benign receiver, through the interference term $\sum_{i=k+1}^K V_a^{(i)}$, quantifies the excess noise variance $(1-T_k)W_k$ induced by Eve's collective attack, and incorporates the detection noise variance $\delta_{\det}^2$.

\subsection{Upper Bound for the Eavesdropping Rate}
To quantify the security of the NOMA-CVQKD system under the collective attack, we analyze the upper bound of the eavesdropping rate at Eve. The lower bound on the key rate (Lemma 2) ensures reliable key generation. The security guarantee requires demonstrating that Eve's accessible information remains limited. 
\subsubsection{Holevo Information Under RR}
With RR, the maximum information extractable by Eve through the collective attack is bounded by the Holevo information $\mathcal{X}(Y_k;E_k)$, which quantifies the quantum information between the BS's measurement outcome $Y_k$ and Eve's intercepted state $E_k$. This is because the BS's heterodyne detection introduces inherent quantum noise, limiting Eve's ability to purify the system.  Eve's intercepted information is further constrained by the BS's publicly announced quadrature measurement choices. The Holevo information $\mathcal{X}(Y_k;E_k)$ is given by \cite{Terahertz-QKD,MIMO2}
\begin{equation}
\label{eq:Holevo}
\mathcal{X} \big( Y_k;E_k \big) = S \big( E_k \big) - S \big( E_k|Y_k \big),
\end{equation}
where $S\left( E_k \right) $ is the von Neumann entropy of the quantum state eavesdropped by Eve, and $S\left( E_k|Y_k \right) $ is the conditional entropy of Eve given the measurement result of the BS.

\subsubsection{Von Neumann Entropy $S\left( E_k \right) $ at Eve}
In the collective attack, Eve interacts with the joint state of the $k$-th user and the BS, $\hat{\rho}_{X_kY_k}$, resulting in a global pure system $\hat{\rho}_{X_kY_kE_k}$.
By the Schmidt decomposition, the reduced subsystems $\hat{\rho}_{E_k}$ and $\hat{\rho}_{X_kY_k}$ from $\hat{\rho}_{X_kY_kE_k}$ share the same set of non-zero eigenvalues and, in turn, the same von Neumann entropy. Computing the entropy $S(E_k)$ of Eve's subsystem reduces to evaluating the von Neumann entropy of the joint state $\hat{\rho}_{X_kY_k}$~\cite{MIMO2,Terahertz-QKD}, yielding
\begin{equation}
\label{eq:V_entropy}
S\big( E_k \big) = h\big( \lambda_{1}^{(k)} \big) + h\big( \lambda_{2}^{(k)} \big),
\end{equation}
where 
\begin{equation}\label{eq:f_entropy}
h\big( x \big) =\frac{x+1}{2}\log _2\frac{x+1}{2}-\frac{x-1}{2}\log _2\frac{x-1}{2}. 
\end{equation}
Moreover, $\lambda _{1}^{\left( k \right)}$ and $\lambda _{2}^{\left( k \right)}$ are the symplectic eigenvalues obtained from the joint Gaussian covariance matrix $\varSigma _{X_kY}$, which determined by evaluating the absolute eigenvalues of the matrix $|i\varOmega \varSigma _{X_kY}|$. Here, $\varSigma _{X_kY}$ characterizes the quantum correlations between the $k$-th user and the BS in the joint state $\hat{\rho}_{X_kY}$ during the CVQKD process, $i$ ensures a real positive eigenvalue, and the symplectic matrix $\varOmega $ is given by~\cite{MIMO2}
\begin{equation}
\label{eq:Omega}
\varOmega =\underset{i=1}{\overset{2}{\oplus}}\left[ \begin{matrix}	0&		1\\	-1&		0\\\end{matrix} \right],
\end{equation}
where $\oplus$ denotes the matrix direct sum operation. $\varSigma_{X_kY}$ can be written as \cite{Terahertz-QKD}
\begin{equation}
\label{eq:Co_ma}
\varSigma_{X_kY} = \left( 
\begin{matrix}
V_{a}^{(k)}\mathbb{I}_2 & \Gamma_k \\
\Gamma_k^\mathsf{T} & b_k\mathbb{I}_2
\end{matrix} 
\right),
\end{equation} 
where $\mathbb{I}_2$ is the two-dimensional identity matrix, $\Gamma_k$ is the coupling matrix given as
\begin{equation}
\Gamma_k = \sqrt{T_k \big( (V_{a}^{(k)})^2 - 1 \big)} \,
\begin{pmatrix}
1 &  0 \\
0 & -1
\end{pmatrix},
\end{equation}
and
\begin{equation}\label{eq: b_k}
b_k=T_kV_{a}^{( k )}+\big ( 1-T_k \big ) W_k+\delta _{\det}^2+V_{\mathrm{I}}.
\end{equation}
Here, the first term \(T_kV_a^{(k)}\) on the right-hand side (RHS) of \eqref{eq: b_k} accounts for the effective signal variance from user~$k$. The second term \((1-T_k)W_k\) quantifies the excess noise variance introduced by Eve' collective attack. The third term \(\delta_{\det}^2\) accounts for detector noise variance.  In the NOMA-CVQKD system, the coherent state of each user is independently Gaussian modulated. The interference variance $V_{\mathrm{I}}$ corresponds to Gaussian-distributed interference from the~$(k+1)$-th to the $K$-th users that are yet to be decoded, and is given by
\begin{equation}
V_{\mathrm{I}}=\sum_{i=k+1}^KT_i{V_{a}^{\left( i \right)}}.
\end{equation}

By substituting \eqref{eq:Omega} and \eqref{eq:Co_ma} into the matrix $|k\varOmega \varSigma _{X_kY}|$, the symplectic eigenvalue $\lambda _{1,2}^{\left( k \right)}$ can be obtained through the following characteristic equation:
\begin{equation}
\label{eq:eigenvalues}
\lambda _{1,2}^{(k)} = \sqrt{\frac{1}{2} \Big( A_k \pm \sqrt{{A_k}^2 - 4B_k} \Big)},
\end{equation}
where
\begin{align}
A_k&={V_{a}^{\left( k \right)}}^2\left( 1-2T_k \right) +2T_k+b_k^2;\\
B_k &= \Big( T_k + \big( 1-T_k \big) V_{a}^{(k)}W_k + V_{a}^{(k)}\big( \delta_{\det}^2 + V_{\mathrm{I}} \big) \Big)^2.
\end{align}
With \eqref{eq:f_entropy}, we can obtain the explicit von Neumann entropy, $S\left( E_k \right) $, at Eve by substituting \eqref{eq:eigenvalues} into~\eqref{eq:V_entropy}.

\subsubsection{Conditional Entropy $S\left( E_k|Y_k \right) $ at Eve}
The conditional entropy \( S(E_k|Y_k) \) of Eve's state conditioned on the BS's measurement depends on the measurement performed by the BS. 
The heterodyne measurements correspond to rank-1 projective measurements, which  collapses the joint pure system $\hat{\rho}_{X_kY_kE_k}$ into a pure conditional system $\hat{\rho}_{X_kE_k|Y_k}$.
Consequently, for the purity of the conditional  system $\hat{\rho}_{X_kE_k|Y_k}$, the conditional von Neumann entropies of subsystems $\hat{\rho}_{E_k|Y_k}$ and $\hat{\rho}_{X_k|Y_k}$ satisfy \( S(E_k|Y_k) = S(X_k|Y_k) \)~\cite{MIMO2}, yielding \cite{Terahertz-QKD}
\begin{equation}
\label{eq:V_con}
S\big( E_k|Y_k \big) = h\big( \lambda _{\mathrm{het}}^{(k)} \big),
\end{equation}
where 
$\lambda _{\mathrm{het}}^{(k)}$ is a conditional symplectic eigenvalue of the $k$-th user’s conditional covariance matrix $\varSigma_{X_k|Y}$ in the RR scenario. With the  homodyne detection at the BS, $\lambda _{\mathrm{het}}^{(k)}$ admits~\cite{MIMO2}
\begin{equation}
\label{eq:eigenvalue_E}
\lambda _{\mathrm{het}}^{(k)} = V_{a}^{(k)} - \frac{T_k \big( {V_{a}^{(k)}}^2 - 1 \big)}{b_k + 1}.
\end{equation}
where $\lambda _{\mathrm{het}}^{(k)}$ accounts for the modulation variance $V_{a}^{(k)}$, the channel transmittance $T_k$, and the noise term $b_k$ for the $k$-th user, $\forall k=1,\cdots,K$.

\subsection{Sum Secret Key Rate and Its Asymptotic Approximation}

By substituting \eqref{eq:V_entropy} and \eqref{eq:V_con} into \eqref{eq:Holevo}, we can obtain the Holevo information \( \mathcal{X}(Y_k;E_k) \) that characterizes the maximum information that Eve can extract about the $k$-th user during the collective attack, as given by 
\begin{equation}
\begin{aligned}
\label{eq:upper}
\mathcal{X} \big( Y_k;E_k \big) &= S\big( E_k \big) - S\big( E_k|Y_k \big) \\
&= h\big( \lambda _{1}^{(k)} \big) + h\big( \lambda _{2}^{(k)} \big) - h\big( \lambda _{\mathrm{het}}^{(k)} \big),
\end{aligned}
\end{equation}
where $\mathcal{X} \big( Y_k;E_k \big)$ characterizes the upper bound of Eve’s accessible information from the $k$-th user.

Given the achievable key rate lower bound $I_{k}^{\left( \mathrm{low} \right)}$ for the $k$-th user in \eqref{eq:I_k_lower} and the upper bound on Eve's information $\mathcal{X} \left( Y_k;E_k \right) $  for the $k$-th channel in \eqref{eq:upper}, the sum SKR at the benign receiver is established in the following theorem. 

\begin{theorem}
In the uplink NOMA-CVQKD system, the sum SKR at the benign receiver is given by
\begin{equation}
\begin{aligned}
\label{eq:SKR}
I_{\mathrm{sum}}^{(\mathrm{sec})}=\sum_{i=1}^K{\left( \eta I_{k}^{\left( \mathrm{low} \right)}-\mathcal{X} \left( Y_k;E_k \right) \right)},
\end{aligned}
\end{equation}
where  $\eta $ denotes the reconciliation efficiency.
The sum SKR converges as the transmit powers of the users increase.
\end{theorem}

\begin{proof}
    See \textbf{Appendix~C}. 
\end{proof}


In the considered 
NOMA-CVQKD system, we can employ the large modulation variance following \cite{MIMO1,Terahertz-QKD}, i.e., \( V_{a}^{(k)} \gg 1,\,\forall k \). This regime improves system robustness against both quantum noise and channel loss.
As numerically validated in \cite[Fig.~3]{MIMO1}, this approximation is accurate when the modulated optical power exceeds $-80$ dBm.
Considering the large modulation \( V_{a}^{(k)} \gg 1 \) in the NOMA-CVQKD system, the symplectic eigenvalues in \eqref{eq:eigenvalues} can be written as  
\begin{align}
\label{eq:eigenvalues_app1}
\lambda _{1}^{(k)} &\xrightarrow{V_{a}^{(k)}\gg 1} \tilde{\lambda}_{1}^{(k)} = \sqrt{{V_{a}^{(k)}}^2 + {b_k}^2},\\
\label{eq:eigenvalues_app2}
\lambda _{2}^{(k)} &\xrightarrow{V_{a}^{(k)}\gg 1} \tilde{\lambda}_{2}^{(k)} = \frac{V_{a}^{(k)}\big( V_{\mathrm{I}} + \delta _{\det}^2 \big)}{\sqrt{{V_{a}^{(k)}}^2 + {b_k}^2}}.
\end{align}
The conditional symplectic eigenvalue of \eqref{eq:eigenvalue_E} is given by 
\begin{equation}
\label{eq:eigenvalues_app3}
\lambda _{\mathrm{het}}^{\left( k \right)}\xrightarrow{V_{a}^{\left( k \right)}\gg 1}\tilde{\lambda}_{\mathrm{het}}^{\left( k \right)}=V_{a}^{\left( k \right)}-\frac{T_k{V_{a}^{\left( k \right)}}^2}{b_k}.
\end{equation}
Moreover, the function of \eqref{eq:f_entropy} admits the asymptotic form
\begin{equation}
\label{eq:asymptotic_entropy_function}
h\left( x \right) \xrightarrow{x\gg 1}\tilde{h}\left( x \right) =\log \left( \frac{ex}{2} \right).
\end{equation}

By substituting \eqref{eq:eigenvalues_app1}, \eqref{eq:eigenvalues_app2}, \eqref{eq:eigenvalues_app3} and \eqref{eq:asymptotic_entropy_function} into \eqref{eq:upper}, we can obtain the asymptotic approximation of the Holevo information $\mathcal{X} \left( Y_k;E_k \right)$, as given by 
\begin{align}
\label{eq:asympt_Holevo}
\mathcal{X}\big( Y_k &;E_k \big) \xrightarrow{V_{a}^{(k)}\gg 1}\tilde{\mathcal{X}}\big( Y_k;E_k \big)\nonumber\\ 
=&\tilde{h}\big( \tilde{\lambda}_{1}^{(k)} \big) +h\big( \tilde{\lambda}_{2}^{(k)} \big) -h\big( \tilde{\lambda}_{\mathrm{het}}^{(k)} \big) \nonumber\\
=&\log \Big( \frac{e}{2}V_{a}^{(k)}\sqrt{{V_{a}^{(k)}}^2+{b_k}^2} \Big)+h\Big( \frac{V_{a}^{(k)}\big( V_{\mathrm{I}}+\delta _{\det}^2 \big)}{\sqrt{{V_{a}^{(k)}}^2+{b_k}^2}} \Big) \nonumber\\
& -h\big( V_{a}^{(k)}-\frac{T_k{V_{a}^{(k)}}^2}{b_k} \big).
\end{align}
By substituting \eqref{eq:asympt_Holevo} into \eqref{eq:SKR}, the asymptotic approximation for the sum SKR is obtained in the following theorem.

\begin{theorem}
In the uplink NOMA-CVQKD system, under large modulation variances, i.e., \( V_{a}^{(k)} \gg 1,\,\forall k=1,\ldots,K \), the asymptotic approximation of sum SKR is given by
\begin{equation}
\begin{aligned}
\label{eq:SKR2}
&I_{\mathrm{sum}}^{(\mathrm{sec})} \xrightarrow{V_{a}^{(k)}\gg 1} \tilde{I}_{\mathrm{sum}}^{(\mathrm{sec})} = \sum_{k = 1}^{K} \Big[ \eta I_{k}^{(\mathrm{low})} - \tilde{\mathcal{X}}\big( Y_k;E_k \big) \Big].
\end{aligned}
\end{equation}
\end{theorem}
This theorem simplifies the sum SKR calculation under the large modulation variance condition (\( V_{a}^{(k)} \gg 1 \)) while preserving analytical tractability, and practicality under large modulation variances \cite{MIMO1,Terahertz-QKD}. It is of practical interest to maximize the sum SKR while maintaining practical transmission rates for secure information and service transmission, as discussed in the following section.

\section{Maximization of Sum SKR for NOMA-CVQKD}
\label{sec:Algorithm}
In this section, we maximize the sum SKR for the uplink NOMA-CVQKD system under a collective attack.
This starts with a problem formulation, followed by problem convexification using the SCA. The optimality, convergence, and computational complexity of the proposed algorithm are discussed.

\subsection{Problem Statement}
The objective of the considered problem is to optimize the variances of the coherent states transmitted by the users to maximize the sum SKR of the uplink NOMA-CVQKD system under a collective attack. This is because the variances directly affect the achievable SKR, as well as the lower bound of the achievable key rate \( I_k^{(\text{low})} \) in \eqref{eq:lower} and the Holevo information \( \mathcal{X}(Y_k; E_k) \) in \eqref{eq:I_k_lower}. The problem can be cast as
\begin{subequations}
\begin{align}
\mathbf{P}1: & \quad \max_{V_{a}^{\left( 1 \right)},\cdots V_{a}^{\left( K \right)}}     \tilde{I}_{\mathrm{sum}}^{\left( \sec \right)} \label{P1_a}\\
&\mathrm{s}.\mathrm{t}.\quad\quad V_{a}^{\left( k \right)}\leqslant V_{\max}^{\left( k \right)}, \forall k=1,\cdots K, \label{P1_b}\\
&\sum_{i=1}^K{T_iV_{a}^{\left( i \right)}}+\sum_{i=1}^K{\left( 1-T_i \right) W_k}+\delta _{\det}^2\leqslant V_{\max}^{\left( \mathrm{BS} \right)}.\label{P1_c}
\end{align}
\end{subequations}
In Problem $\mathbf{P}1$, the optimization variables are the modulation variances $V_{a}^{\left( k \right)}, \forall k=1,\cdots K$ of the transmitted coherent state $|\alpha _{p_k,q_k}^{\left(k\right)}\rangle$ sent by each user $k$, subject to a maximum transmit power limit $V_{\max}^{\left( k \right)}$, i.e., \eqref{P1_b}. At the receiver, there is a maximum received power limit $V_{\max}^{\left( \mathrm{BS} \right)}$, i.e.,  \eqref{P1_c}, due to saturation effects \cite{Yongkang2}.

Unfortunately, the objective function $I_{\mathrm{sum}}^{\left( \sec \right)}$ in Problem $\mathbf{P}1$ presents significant analytical intractability.
Specifically, the lower bound on the achievable key rate of any user $k$, $I_{k}^{\left( \mathrm{low} \right)}$, exhibits concavity as it is a logarithmic function over the feasible domain, while the Holevo information, $\mathcal{X} \left( Y_k;E_k \right)$, is neither convex nor concave.  
In the following subsections, we employ SCA to convexify the objective function, thereby enabling efficient power allocation among users while maintaining the problem's security constraints.

\subsection{SCA-Based Power Allocation Algorithm}
We resort to the SCA to convexify the asymptotic approximation Holevo information $\tilde{\mathcal{X}}( Y_k;E_k) $ in \eqref{eq:SKR2}, where $h( \tilde\lambda _{2}^{( k)})$ exhibits concavity and $-h( \tilde\lambda _{\mathrm{het}}^{( k)} )$ exhibits convexity with the feasible domain of $V_{a}^{( k )}>0$, $b_k>0$ and $ 0<T_k<1$. On the other hand, $\tilde{h}(\tilde{\lambda}_{1}^{(k)})$ exhibits convexity for $V_a^{(k)} < b_k$ and concavity for $V_a^{(k)} > b_k$, with a critical inflection at $V_a^{(k)} = b_k$; see \textbf{Appendix~D}.

We construct a surrogate concave function of  $\tilde{h}( \tilde\lambda _{1}^{( k )} )$ for $V_a^{(k)} > b_k$ and a surrogate concave function of $h( \tilde\lambda _{2}^{( k)})$ in the feasible domain $V_{a}^{( k )}>0$, $b_k>0$ and $ 0<T_k<1$. Consequently, we obtain a surrogate concave function for the Holevo information $\tilde{\mathcal{X}}( Y_k;E_k) $. 
Problem \textbf{P1} can be convexified since the lower bound \(I_{k}^{(\mathrm{low})}\) in the objective function is concave and constraints  \eqref{P1_b} and \eqref{P1_c} are linear and affine.

We first construct the surrogate concave functions of $\tilde{h}( \tilde{\lambda}_{1}^{( k )} ) $ for $V_a^{(k)} > b_k$.
Applying the first-order Taylor series approximation to expand  $\tilde{h}( \tilde{\lambda}_{1}^{( k )} ) $ at the $t$-th SCA local point yields 
\begin{equation}
\tilde{\lambda}_{1}^{(k),(t-1)}=V_{a}^{(k),(t-1)}\sqrt{\big( V_{a}^{(k),(t-1)} \big) ^2+\big( {b_k}^{(t-1)} \big) ^2},
\end{equation}
where ${b_k}^{\left( t-1 \right)}=T_kV_{a}^{\left( k \right) ,\left( t-1 \right)}+\left( 1-T_k \right) W+\delta _{\det}^2+V_{\mathrm{I}}$.
Then, the surrogate concave function of $\tilde{h}( \tilde{\lambda}_{1}^{\left( k \right)}) $ for the $t$-th SCA iteration is given by
\begin{equation}
\begin{aligned}
\label{eq:sca_lam1}
\tilde{h}_{\mathrm{SCA}}^{(t)}\big( \tilde{\lambda}_{1}^{(k)},\tilde{\lambda}_{1}^{(k),(t-1)}\big) = \tilde{h}_{\mathrm{SCA}}^{(t)}\big( \tilde{\lambda}_{1}^{(k),(t-1)} \big) \\ 
+ \tilde{h}_{\mathrm{SCA}}^{(t)\prime}\big( \tilde{\lambda}_{1}^{(k),(t-1)} \big) 
\big( \tilde{\lambda}_{1}^{(k)}-\tilde{\lambda}_{1}^{(k),(t-1)} \big),
\end{aligned}
\end{equation}
where 
\begin{equation}
\begin{aligned}
\label{eq:sec_h_lam1}
&\tilde{h}_{\mathrm{SCA}}^{(t)}( \tilde{\lambda}_{1}^{(k),(t\!-\!1)} )
\!=\!\log_2\!\big(\frac{e}{2}\! V_{a}^{(k),(t\!-\!1)}\!\!\sqrt{( V_{a}^{(k),(t\!-\!1)} ) ^2\!\!\!+\!( {b_k}^{(t\!-\!1)} ) ^2} \big),
\end{aligned}
\end{equation}
and
\begin{equation}
\begin{aligned}
\label{eq:sec_h_lam12}
\tilde{h}_{\mathrm{SCA}}^{(t)\prime}( \tilde{\lambda}_{1}^{(k),(t\!-\!1)})\!= \!&\frac{1}{V_{a}^{(k),(t\!-\!1)}\!\ln 2}\!\!+\!\!
\frac{V_{a}^{(k),(t-1)}+{b_k}^{(t-1)}T_k}{( V_{a}^{(k),(t\!-\!1)} )^2\!+\!( {b_k}^{(t\!-\!1)} )^2}\frac{1}{\ln 2}.
\end{aligned}
\end{equation}

Next, we construct the surrogate concave function of  ${h}( \tilde{\lambda}_{2}^{( k )}) $ by applying the first-order Taylor series approximation to expand ${h}( \tilde{\lambda}_{2}^{( k)} ) $ at the $t$-th SCA local point, as given by  
\begin{equation}
\tilde{\lambda}_{2}^{(k),(t-1)}=\frac{V_{a}^{(k),(t-1)}\big( V_{\mathrm{I}}+\delta _{\det}^2 \big)}{\sqrt{\big( V_{a}^{(k),(t-1)} \big)^2+\big( {b_k}^{(t-1)} \big)^2}}.
\end{equation}
Then, the surrogate convex function of ${h}( \tilde{\lambda}_{2}^{( k)}) $ in the $t$-th SCA iteration is given by
\begin{align}
\label{eq:sca_lam2}
{h}_{\mathrm{SCA}}^{(t)}( \tilde{\lambda}_{2}^{(k)},\tilde{\lambda}_{2}^{(k),(t-1)} ) =&{h}_{\mathrm{SCA}}^{(t)}( \tilde{\lambda}_{2}^{(k),(t-1)} ) +{h}_{\mathrm{SCA}}^{\prime(t)}( \tilde{\lambda}_{2}^{(k),(t-1)} ) \nonumber
\\
&
\times \big( \tilde{\lambda}_{2}^{(k)}-\tilde{\lambda}_{2}^{(k),(t-1)}\big),
\end{align}
where
\begin{align}
\label{eq:sec_h_lam2}
{h}_{\mathrm{SCA}}^{(t)}\big( \tilde{\lambda}_{2}^{(k),(t-1)} \big) 
&= \frac{\tilde{\lambda}_{2}^{(k),(t-1)}\!\!+\!1}{2}\log_2 \Big( \frac{\tilde{\lambda}_{2}^{(k),(t-1)}\!\!+\!\!1}{2} \Big)- \nonumber\\
&\quad  \frac{\tilde{\lambda}_{2}^{(k),(t-1)}\!\!\!-\!\!1}{2}\log_2 \Big( \frac{\tilde{\lambda}_{2}^{(k),(t-1)}\!\!\!-\!\!1}{2} \Big),
\end{align}
and
\begin{align}
\label{eq:sec_h_lam22}
{h}_{\mathrm{SCA}}^{\prime(t)}&\big( \tilde{\lambda}_{2}^{(k),(t-1)} \big) 
= \frac{1}{2}\log \bigg( \frac{\tilde{\lambda}_{2}^{(k),(t-1)}+1}{\tilde{\lambda}_{2}^{(k),(t-1)}-1} \bigg) \times \notag\\
&\frac{\big( V_{\mathrm{I}}+\delta _{\det}^2 \big) \big( \big( {b_k}^{(t-1)} \big)^2-V_{a}^{(k),(t-1)}{b_k}^{(t-1)}T_k \big)}{\Big[ \big( V_{a}^{(k),(t-1)} \big)^2+\big( {b_k}^{(t-1)} \big) \Big]^{3/2}} \times\notag\\
&\big( \tilde{\lambda}_{2}^{(k)}-\tilde{\lambda}_{2}^{(k),(t-1)} \big).
\end{align}

Substituting the surrogate concave functions \eqref{eq:sca_lam1} and \eqref{eq:sca_lam2} into \eqref{eq:asympt_Holevo}, we can obtain the surrogate concave function of the asymptotic approximation of the Holevo information, $\tilde{\mathcal{X}}\left( Y_k;E_k \right) $, as given by
\begin{equation}
\label{eq:asympt_Holevo_sca}
\tilde{\mathcal{X}}_{\mathrm{SCA}}^{(t)}\big( Y_k;E_k \big)\! =\! \mathcal{H}_1^{(t)}( \tilde{\lambda}_1^{(k)} )\! + \!{h}_{\mathrm{SCA}}^{( t )}( \tilde{\lambda}_{2}^{( k)},\tilde{\lambda}_{2}^{( k ) ,( t\!-\!1 )}) \!- \!{h}\big( \tilde{\lambda}_{\mathrm{het}}^{(k)} \big),
\end{equation}
where 
\begin{equation}
\begin{aligned}
\mathcal{H}_1^{(t)}( \tilde{\lambda}_1^{(k)} ) = 
\begin{cases} 
\tilde{h}_{\mathrm{SCA}}^{( t )}( \tilde{\lambda}_{1}^{( k)},\tilde{\lambda}_{1}^{( k ) ,(t-1)}), & \text{if } V_a^{(k)} > b_k; \\
\tilde{h}\big( \tilde{\lambda}_{1}^{(k)} \big), &  \text{if } V_a^{(k)} \leq b_k.
\end{cases}
\end{aligned}
\end{equation}
The convexity of $\tilde{\mathcal{X}}_{\rm SCA}( Y_k;E_k ) $ is confirmed in the feasible domain of $V_{a}^{( k )}>0$, $b_k>0$ and $ 0<T_k<1$.

By substituting \eqref{eq:asympt_Holevo_sca} into \eqref{eq:SKR2}, we obtain the asymptotic  sum SKR in the $t$-th SCA iteration, as given by
\begin{equation}
\begin{aligned}
\label{eq:SKR_SCA}
\tilde{I}_{\mathrm{sum},\mathrm{SCA}}^{(\sec),(t)}=\sum_{k=1}^K\Big[ \eta I_{\mathrm{k}}^{(\mathrm{low})}-\tilde{\mathcal{X}}_{\mathrm{SCA}}^{(t)}\big( Y_k;E_k \big) \Big].
\end{aligned}
\end{equation}
Problem \textbf{P1} is eventually reformulated as 
\begin{subequations}
\begin{align}
\mathbf{P}2: & \quad \max_{V_{a}^{\left( 1 \right)},\cdots, V_{a}^{\left( K \right)}} \,\,    \tilde{I}_{\mathrm{sum},\mathrm{SCA}}^{\left( \sec \right) ,\left( t \right)}\\
&\mathrm{s}.\mathrm{t}. \quad\eqref{P1_b},\eqref{P1_c},\notag
\end{align}
\end{subequations}
which maximizes the convexified asymptotic sum SKR $\tilde{I}_{\mathrm{sum},\mathrm{SCA}}^{\left( \sec \right),\left( t \right)}$ in \eqref{eq:SKR_SCA},
and is constrained by linear power constraints \eqref{P1_b} and \eqref{P1_c}, forming a convex feasible set. Consequently, Problem \textbf{P2} constitutes a convex problem that can be efficiently solved using the CVX toolbox. 
The solution obtained from the CVX toolbox in the $t$-th SCA iteration serves as the local point for the $(t+1)$-th SCA iteration. This updated local point is utilized to formulate the convex Problem~\textbf{P2} for the $(t+1)$-th SCA iteration. Upon the convergence of the SCA iterations, the power allocation solution is achieved.

\begin{algorithm}[t]
\caption{Power allocation for uplink NOMA-CVQKD}\label{alg:alg1}
\begin{algorithmic}
\item \hspace{0.35cm} \textbf{1: Input:}$~K, T_k, W_k, \delta_{\det}^2, \eta, V_{\max}^{(k)}, V_{\max}^{(\mathrm{BS})}, T_{\max}, $
\item \hspace{1.8cm} $\tau_{\mathrm{SCA}}, \forall k=1,\cdots,K;$
\item \hspace{0.35cm} \textbf{2: Initialization:}$~t=0;~V_{a}^{(k),(0)}, \forall k=1,\cdots,K;$
\item \hspace{0.35cm}\textbf{3: }\text{~Compute initial sum SKR $\tilde{I}_{\mathrm{sum},\mathrm{SCA}}^{(\sec),(0)}$};
\item \hspace{0.5cm}$  \textbf{4: while}~\left|\tilde{I}_{\mathrm{sum},\mathrm{SCA}}^{(\sec),(t)} -\tilde{I}_{\mathrm{sum},\mathrm{SCA}}^{(\sec),(t-1)}\right| > \tau_{\mathrm{SCA}}$  or 
\item \hspace{1.5cm}$ ~~~~t<T_{\max};$ {for} $k=1$ to $K$; \textbf{do}
\item \hspace{0.5cm}{\textbf{5: }\text{~~~Solving Problem \textbf{P2} using the CVX toolbox;} 
\item \hspace{0.5cm}\textbf{6: }\text{~~~Update local points: $\tilde{\lambda}_1^{(k),(t-1)}$ and $\tilde{\lambda}_2^{(k),(t-1)}$};

\item \hspace{0.5cm}$ \textbf{7: ~~~$t=t+1;$}$
\item \hspace{0.5cm}\textbf{8: }\text{~~~Compute $\tilde{I}_{\mathrm{sum},\mathrm{SCA}}^{(\sec),(t)}$ with new $V_{a}^{(k),(t)}$};
\item \hspace{0.5cm}$ \textbf{9: end while;}$
\item \hspace{0.35cm} \textbf{10: Output:} $V_{a}^{(k)} = V_{a}^{(k),(t)}, \forall k=1,\cdots,K;$
}
\end{algorithmic}
\end{algorithm}

\subsection{Convergence and Complexity Analysis}
\label{Convergence Analysis: Algorithm 1}
\textbf{Algorithm 1} presents the proposed power allocation scheme for the uplink NOMA-CVQKD system under collective attacks. It jointly optimizes the modulation variances \( V_a^{(k)} ,\forall k\) of the coherent states of the users.
The algorithm iteratively solves Problem~\textbf{P2} in Steps 4 to 9. Particularly, Step 5 solves Problem~\textbf{P2} using the CVX toolbox. Then, Step 6 updates the local point for the next SCA iteration. 

\subsubsection{Convergence Analysis}
According to \cite{kkt1} and \cite{kkt2}, \textbf{Algorithm 1} converges to a solution that satisfies the KKT conditions of Problem \textbf{P1}. Specifically, for $V_a^{(k)} > b_k$, the surrogate concave function $\tilde{h}_{\mathrm{SCA}}^{( t )}( \tilde{\lambda}_{1}^{( k )},\tilde{\lambda}_{1}^{( k ) ,( t-1 )} ) $ in \eqref{eq:sca_lam1} and the corresponding original function $\tilde{h}( \tilde{\lambda}_{1}^{( k )} )$ take the same value, $\log_2\!\Big( \!\frac{e}{2} V_{a}^{(k),(t-1)}\!\!\sqrt{\big( V_{a}^{(k),(t-1)} \big) ^2\!+\!\big( {b_k}^{(t-1)} \big) ^2} \Big)$,  and the same gradient, $\frac{1}{V_{a}^{(k),(t-1)}\!\ln 2}\!+\!\frac{V_{a}^{(k),(t-1)}+{b_k}^{(t-1)}T_k}{( V_{a}^{(k),(t-1)} )^2\!+\!( {b_k}^{(t-1)} )^2}\frac{1}{\ln 2}$, at the local point $\tilde{\lambda}_{1}^{( k )}=\tilde{\lambda}_{1}^{( k ) ,( t-1 )}$. In the feasible domain, $\tilde{h}( \tilde{\lambda}_{1}^{( k )} ) \geqslant \tilde{h}_{\mathrm{SCA}}^{( t)}( \tilde{\lambda}_{1}^{( k)},\tilde{\lambda}_{1}^{( k ) ,( t-1 )}) $.
For $V_a^{(k)} \leq b_k$, $\tilde{h}\big( \tilde{\lambda}_{1}^{(k)} \big)$ exhibits convexity and, thus preserves the convexity of Problem~\textbf{P2}.

Moreover, given \( T_k, W_k, \) and \( \delta_{\text{det}}^2 \) in \eqref{eq: b_k}, we have $\frac{\partial b_k}{\partial V_a^{(k)}} = T_k < 1$. In other words, \( b_k \) grows slower than \( V_a^{(k)} \) (since \(0< T_k < 1 \)). If \( V_a^{(k),(t)} > b_k^{(t)} \) in any SCA iteration $t$, then \( V_a^{(k)} > b_k \)  for all subsequent iterations.

Similarly, the surrogate convex function ${h}_{\mathrm{SCA}}^{( t)}( \tilde{\lambda}_{2}^{( k )},\tilde{\lambda}_{2}^{( k ) ,( t-1 )} ) $ in \eqref{eq:sca_lam2} and
the corresponding original function ${h}( \tilde{\lambda}_{2}^{( k )} ) $ share the same value, $\frac{\tilde{\lambda}_{2}^{(k),(t\!-\!1)}\!\!+\!1}{2}\!\log_2 \!\big( \frac{\tilde{\lambda}_{2}^{(k),(t\!-1\!)}\!\!+\!1}{2} \big)\!\!- \!\! \frac{\tilde{\lambda}_{2}^{(k),(t\!-\!1)}\!\!-\!1}{2}\!\log_2 \!\big( \frac{\tilde{\lambda}_{2}^{(k),(t\!-\!1)}\!-\!1}{2} \big)$, and the same gradient, $\frac{1}{2}\log \big( \frac{\tilde{\lambda}_{2}^{(k),(t-1)}+1}{\tilde{\lambda}_{2}^{(k),(t-1)}-1} \Big)  \big( \tilde{\lambda}_{2}^{(k)}-\tilde{\lambda}_{2}^{(k),(t-1)} \big)\times \frac{\big( V_{\mathrm{I}}+\delta _{\det}^2 \big) \big( \big( {b_k}^{(t-1)} \big)^2-V_{a}^{(k),(t-1)}{b_k}^{(t-1)}T_k \big)}{\big[ ( V_{a}^{(k),(t-1)} )^2+( {b_k}^{(t-1)} ) \big]^{3/2}} $, at the local point $\tilde{\lambda}_{2}^{( k )}=\tilde{\lambda}_{2}^{( k) ,( t-1 )}$. Within the feasible domain, the condition ${h}( \tilde{\lambda}_{2}^{( k )} ) \geqslant {h}_{\mathrm{SCA}}^{( t )}( \tilde{\lambda}_{2}^{( k)},\tilde{\lambda}_{2}^{( k ) ,( t-1 )} ) $ holds. 

In general, Problem~\textbf{P2} is convex and hence satisfies Slater’s condition at the $t$-th iteration of the SCA. Its solution is a KKT stationary point of the problem. According to \cite[Thm 1]{kkt1}, the solution of \textbf{Algorithm 1} also converges to a KKT stationary point of Problem \textbf{P1}. Specifically, the surrogate functions maintain gradient consistency with the original functions, and their bounding properties. Additionally, in Problem~\textbf{P2}, the convergence properties hold across the operational regimes: For $V_a^{(k)} > b_k$, the surrogate approximations preserve local optimality; for $V_a^{(k)} \leq b_k$, the inherent convexity of the original functions ensures the optimal solutions. Thus, the solution lies in the interior of the feasible region of Problem~\textbf{P2}. These properties guarantee that the obtained KKT point constitutes a local optimum of the original Problem \textbf{P1}.

\subsubsection{Complexity Analysis}
\textbf{Algorithm 1} solves the convex Problem \textbf{P2} in each SCA iteration using the interior point method. The computational complexity of the interior point method is $\mathcal{O}(\max\{N_c,N_v\}^4\sqrt{N_v})$ for $N_v$ variables and $N_c$ constraints~\cite{opt_variables,Convex}. We have $N_v = K$ optimization variables (i.e., the variance $V_a^{(k)}$ of the transmitted coherent state for user $k$) and $N_c = K+1$ constraints (i.e., $K$ transmit power constraints on the transmitters, and a power constraint on the receiver), resulting in a per-iteration complexity of $\mathcal{O}(K^{4.5})$. 
The evaluation of the objective function involves calculating Holevo information $\mathcal{X}(Y_k;E_k)$ for each user, which requires computing symplectic eigenvalues of $2\times2$ covariance matrices $\Sigma_{AB}^{(k)}$, $\forall k=1,\cdots,K$, with additional complexity $\mathcal{O}(K)$ per iteration. Since the SCA converges within $\mathcal{O}(\log(1/\tau_{\mathrm{SCA}}))$ iterations for accuracy $\tau_{\mathrm{SCA}} \in (0,1)$ \cite{kkt1}, the overall computational complexity is $\mathcal{O}(K^{4.5}\log(1/\tau_{\mathrm{SCA}}))$.

\section{Simulations and Numerical Results}
\label{sec:Simulation}
Extensive simulations are conducted to evaluate the studied NOMA-CVQKD system and the proposed power allocation scheme, i.e., \textbf{Algorithm~1}. 

\subsection{Simulation Setup}
The effects of atmospheric turbulence, the noise at the heterodyne receiver, and the interference between users are considered. 
The coherent states in atmospheric turbulence fluctuate, following the log-normal distribution due to scintillation phenomena \cite{quantum_log}. The PDF of the turbulence is 
$\mathrm{Pr}\big( T_{k,t} \big) =\frac{1}{T_{k,t}\sqrt{2{\pi \sigma _x}^2}}\exp \big[ -\frac{1}{2{\sigma _x}^2}\big( \ln T_{k,t}+\frac{{\sigma _x}^2}{2} \big) ^2 \big] $, where $\sigma _x$ is the intensity of turbulence \cite{quantum_log} and $T_{k,l}=\frac{1}{{d_k}^2}\left( \frac{\pi D_TD_R}{2\nu} \right) ^2$. Here, $d_k$ is the distance between the BS and user $k$; $\nu$ is the wavelength; $D_T$ is the transmitter aperture diameter; $D_R$ is the receiver aperture diameter~\cite{Binary}. The simulated communication distance ranges from 50 m to 200 m, encompassing practical CVQKD applications, including urban inter-building links and near-ground communications \cite{Practical-distance1,Practical-distance2}.
The power constraint \( V_{\max}^{(\mathrm{BS})} \) in \eqref{P1_c} is mapped to the optical power limit \( P_{\max}^{(\mathrm{BS})} \) in dBm, following the standard conversion~\cite{power_dbm}. Unless otherwise stated, simulation parameters are provided in Table~\ref{tab:simulation}.

\begin{table}[t]
  \centering
  \renewcommand{\arraystretch}{1}
  \caption{Simulation Parameters for NOMA-CVQKD Systems}
  \label{tab:simulation}
  \begin{tabular}{p{0.13\linewidth}p{0.45\linewidth}l}
    \hline
    \textbf{Parameter} & \textbf{Description} & \textbf{Typical Value} \\
    \hline
    $K$ & Total number of users & 16 \\
    $\eta$ & Reconciliation efficiency & 0.92\\
    $\tau_d$ & Detector quantum efficiency & 0.6 \\
    $d_k$ & Link distance & 50 $\sim$ 200 m \\
    $\lambda$ & Wavelength & 1550 nm \\
    $D_T$& Transmitter aperture diameter & $10$ cm\\
    $D_R$& Receiver aperture diameter & $1$ m\\
    $\delta_{\det}$ & Heterodyne receiver noise variance & 0.16 \cite{Het_noise}\\
    $\nu$ & Transmission wavelength & $1550$ nm \cite{quantum_log}\\
    $\sigma _x$& Turbulence intensity & $0.3\sim0.6$ \\
    $W_k$ & Excess noise variance & $0.01\sim0.2$ \\
    \hline
  \end{tabular}
\end{table}

\subsection{Benchmarks}
    To the best of our knowledge, no existing work is directly comparable with \textbf{Algorithm 1} due to our study of the new uplink NOMA-CVQKD  system. 
    The following alternative approaches to \textbf{Algorithm 1} are constructed to serve as the benchmarks for the algorithm: 
    \begin{itemize}
    
    \item {\em Quantum OMA (Q-OMA)}: This algorithm allocates orthogonal resources to different users. Each user transmits an independent Gaussian-modulated coherent state during a CVQKD process. The power allocated to user \(k\) is 
    \(V_a^{(k)}= V_{\max}^{\left( \mathrm{BS} \right)}  \frac{T_k}{\sum_{i=1}^K T_i}\) \cite{OFDMA_QKD}. This approach offers a performance baseline for NOMA-CVQKD.
    
    \item {\em Uniform Quantum Power Allocation (UQPA)}: This algorithm assigns equal modulation variance to all users: $V_a^{\left( 1 \right)}=\cdots = V_a^{\left( K \right)}$.
    It is a standard and implementation-friendly 
    strategy of NOMA-CVQKD~\cite{UQPA}. 
    
    \item {\em Conservative Interference Handling (CIH)}: 
    Following the multi-user interference analysis in \cite{CIH}, when determining the SKR for the $k$-th user, this algorithm treats the signals from the other $(K-1)$ users as noise at the heterodyne receiver. The algorithm represents a fundamental performance lower bound for NOMA-CVQKD. 
    \end{itemize}

\subsection{Performance Analysis and Discussion}
\begin{figure}
\centering
\includegraphics[width=3in]{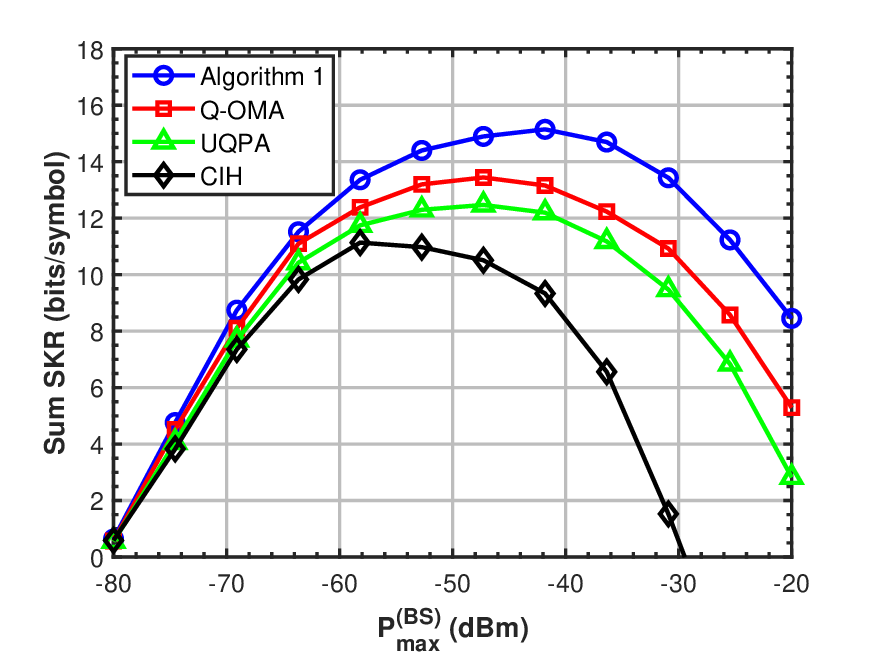}
\caption{Comparison of sum SKR between \textbf{Algorithms 1}, Q-OMA, UQPA, and CIH under 16 users.}
\label{SKR_16users}
\end{figure}
Fig.~\ref{SKR_16users} evaluates the sum SKR of \textbf{Algorithm~1} in comparison with the three benchmarks. When the received power \( P_{\max}^{(\mathrm{BS})} <-60~\text{dBm}\), the sum SKR of all four schemes gradually increases. This is because, at lower transmit power levels, the growth rate of the lower bound of the legitimate channel capacity $I_{k}^{\left( \mathrm{low} \right)}$ exceeds the information \( \mathcal{X}(Y_k; E_k) \) acquired by Eve. As the transmit power increases, e.g., \( P_{\max}^{(\mathrm{BS})} >-40~\text{dBm}\), the sum SKR decreases. 
This is because Eve’s extractable information grows with the transmit power under a collective attack, where Eve can store quantum states and perform joint measurements.
It is evident that \textbf{Algorithm~1} allows multiple users to use the same spectrum resources and SIC to progressively decode users' signals, achieving a higher sum SKR and faster growth than the benchmarks.

\begin{figure}
\centering
\includegraphics[width=3.5in]{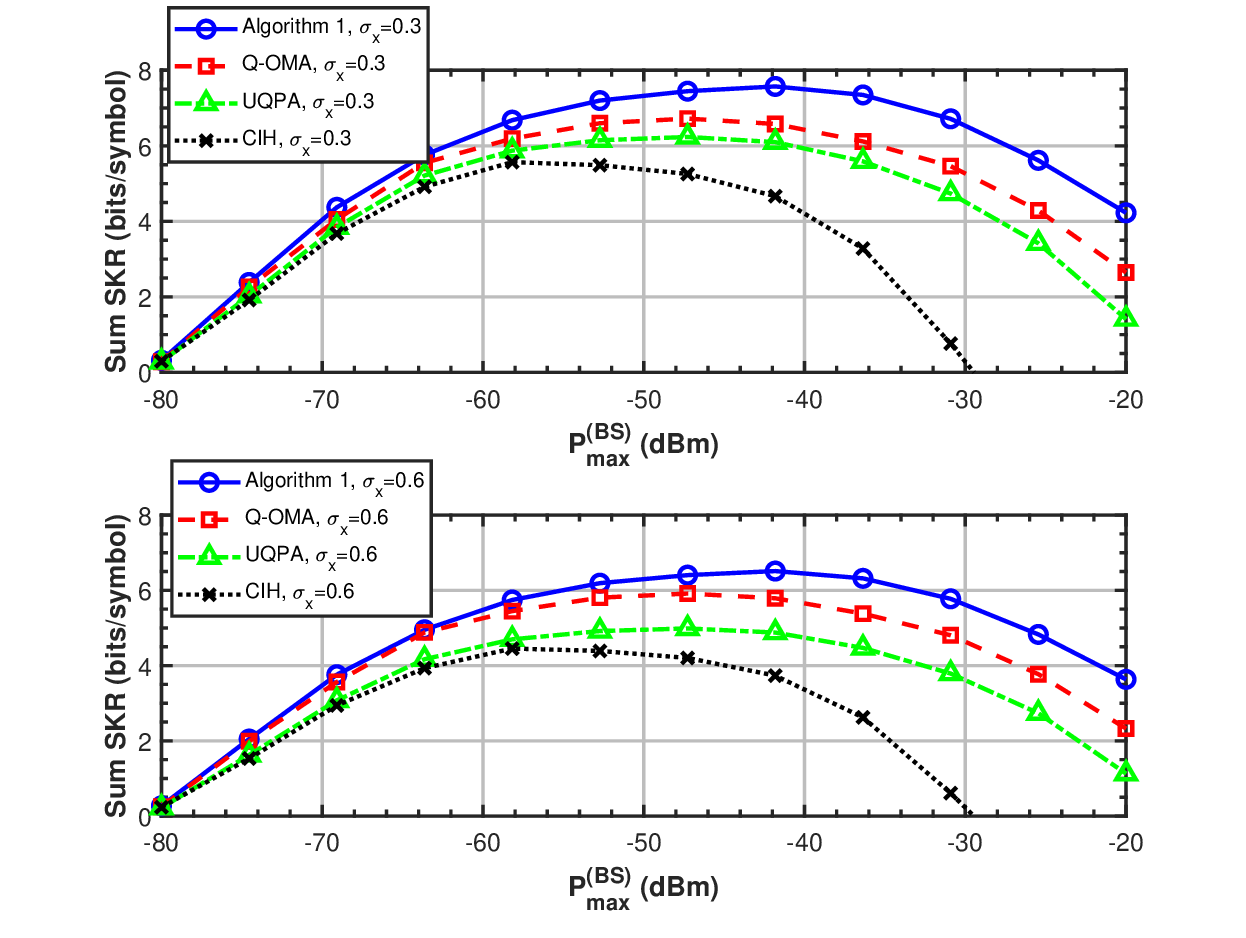}
\caption{ Comparison of sum SKR under different atmospheric turbulence scenarios, where there are 8 users.}
\label{Turbulence_intensity}
\end{figure}
Fig.~\ref{Turbulence_intensity} plots the sum SKR of the uplink NOMA-CVQKD system 
under varying turbulence channel conditions. The uplink NOMA-CVQKD system, in coupling with \textbf{Algorithm~1}, achieves a consistently higher sum SKR than the benchmarks, i.e., Q-OMA, UQPA, and CIH. 
The NOMA-CVQKD scheme mitigates turbulence-induced fading by dynamically adjusting the SIC decoding order and power allocation among users with \textbf{Algorithm~1}, ensuring minimal performance degradation even under different turbulence conditions.
By contrast, all benchmarks suffer from rigid resource partitioning, making them susceptible to fading in turbulence.

\begin{figure}
\centering
\includegraphics[width=3.5in]{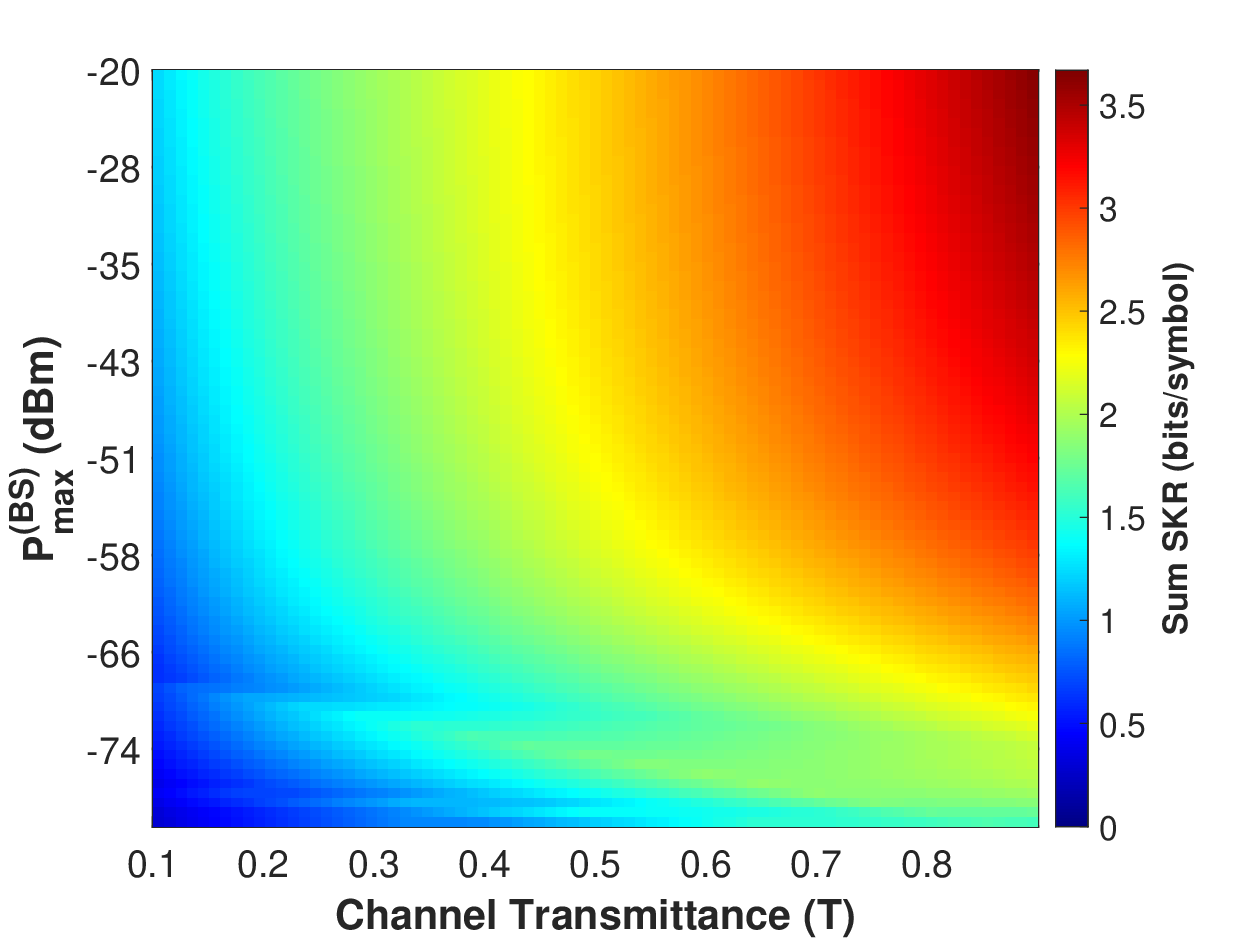}
\caption{Sun SKR for different received power and channel transmittance.}
\label{heat_map}
\end{figure}
Fig.~\ref{heat_map} demonstrates the joint impact of received signal power and channel transmittance on the sum SKR of a 4-user NOMA-CVQKD system. In order to gauge the effect of channel transmittance on the received power, we assume identical transmittance across all users ($T_1=\cdots=T_4=T$), with the SIC decoding order determined by their transmit powers.
In the low-transmittance regime with $T<0.2$, as the received signal power increases, the improvement in sum SKR remains negligible. This is because low transmittance severely limits the signal reception at the receiver, causing the legitimate channel capacity to be primarily constrained by channel loss and excess noise. In this regime, Eve's collective attack advantage becomes prominent as Eve can exploit channel loss purification to extract information.

Conversely, in the high-transmittance regime with $T>0.8$, where the received signal power \( P_{\max}^{(\mathrm{BS})}\) exceeds $-55$ dBm, the sum SKR reaches over 3 bits/symbol. This performance enhancement occurs because the system reaches quantum-noise-limited operation, where the signal power exceeds both detection noise and interference.
The performance transition observed at $T=0.5$ reflects a critical security threshold in the NOMA-CVQKD system. When $T>0.5$, the system effectively resists Eve's collective attack, while its security degrades significantly when $T<0.5$.

\begin{figure}
\centering
\includegraphics[width=3.5in]{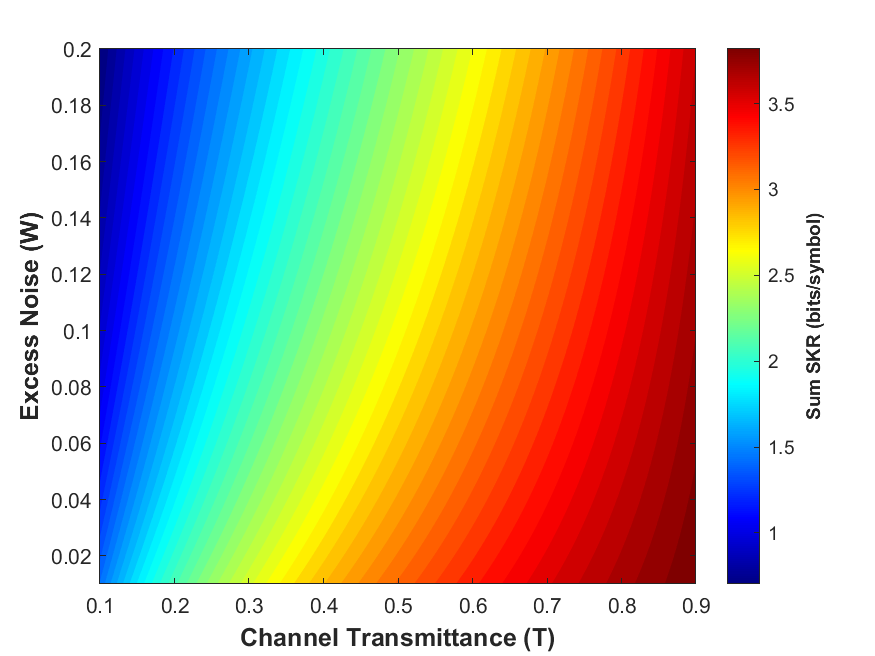}
\caption{Sun SKR for different excess noises and channel transmittance.}
\label{heat_map_TW}
\end{figure}

Fig.~\ref{heat_map_TW} shows the joint impact of excess noise and channel transmittance on the sum SKR of the 4-user NOMA-CVQKD system. To assess how the excess noise affects the system's security performance under varying transmittance, we assume all 4 users operate under the same excess noise $W$ and channel transmittance $T$ conditions. 
It is observed that the sum SKR exhibits a nonlinear relationship with channel transmittance. 
When $T > 0.7$, the system exhibits strong tolerance to excess noise, sustaining a high SKR of 3.5 bits/symbol even at $W = 0.2$. However, when $T < 0.2$, the performance degrades significantly, with the SKR dropping below 1.5 bits/symbol. 
To this end, higher channel transmittance enhances robustness against excess noise in the NOMA-CVQKD system.  

\begin{figure}
\centering
\includegraphics[width=3.5in]{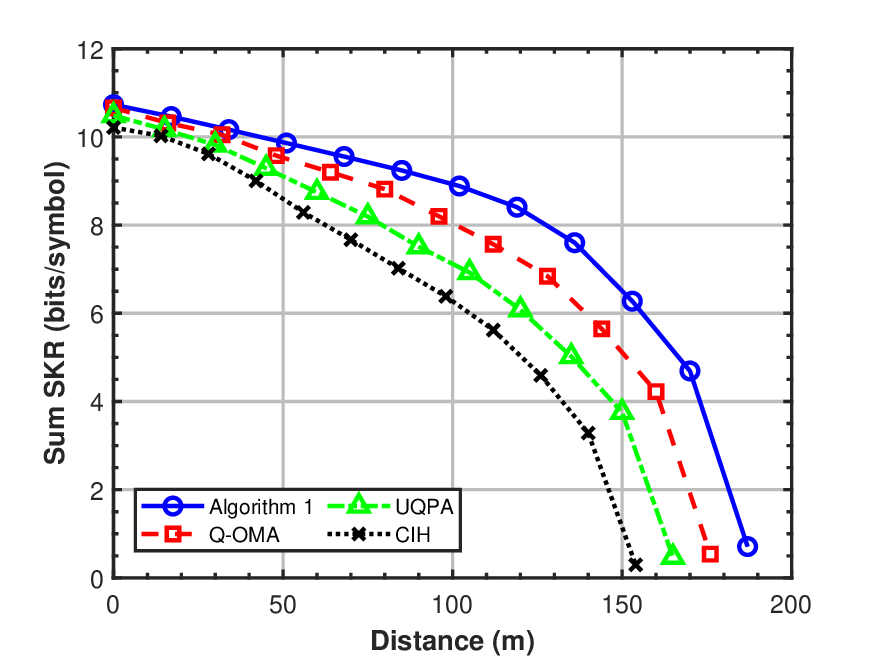}
\caption{ Comparison of sum SKR with the increasing communication distance, where there are 12 users.}
\label{Distance}
\end{figure}

Fig.~\ref{Distance} shows the variation of the sum SKR under \textbf{Algorithm~1} and the benchmarks in a 12-user scenario, as the communication distance increases. 
The sum SKR of all four schemes gradually decreases with the growth of the communication distance. 
Within 140 meters, the sum SKR declines slowly, primarily limited by the detector noise and inherent system noise. 
Beyond 140 meters, the sum SKR drops sharply due to channel attenuation and the increasing excess noise introduced by Eve.  
It is shown that \textbf{Algorithm~1} maintains a higher sum SKR and a slower decay than the benchmarks. This stems from its SIC-based dynamic decoding order and power allocation, which allows users with poorer channel conditions to benefit from the interference cancellation of those with better channels. 
Furthermore, with the hardware configuration \(D_T = 10\, \text{cm}\), \(D_R = 1\, \text{m}\), and \(\lambda = 1550\, \text{nm}\), the system maintains a sustainable positive sum SKR for transmission distances up to 200 m. 
In this sense, \textbf{Algorithm~1} not only outperforms performance metrics but also exhibits superior adaptability to varying distances, offering practical value in real-world deployments.

\begin{figure}
\centering
\includegraphics[width=3.5in]{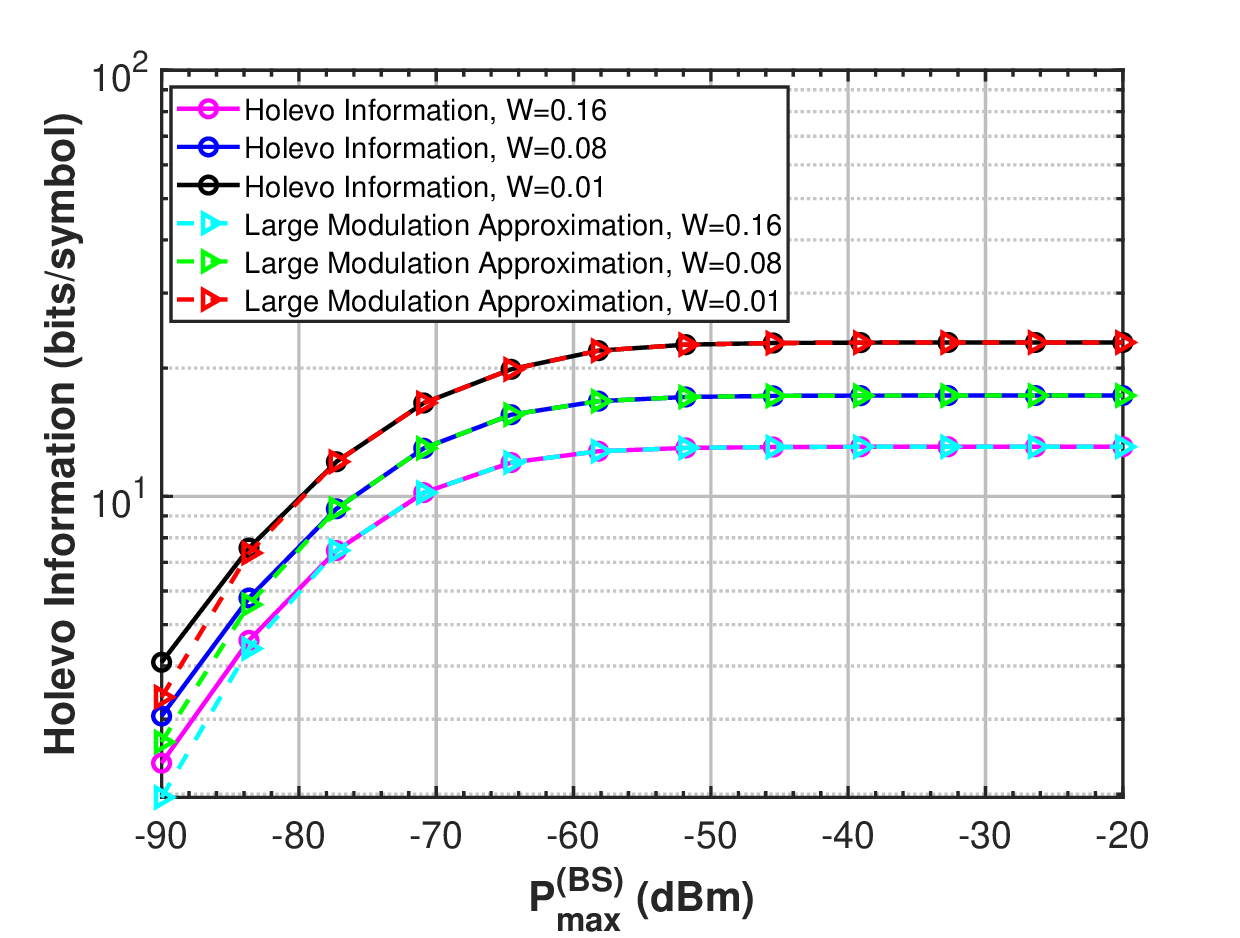}
\caption{Comparison of explicit Holevo information and its large modulation approximation in the NOMA-CVQKD system with 12 users.}
\label{Approximation}
\end{figure}

Fig.~\ref{Approximation} simulates the explicit Holevo information of \eqref{eq:upper} with its large modulation approximation of \eqref{eq:asympt_Holevo} in a 12-user NOMA-CVQKD system under various excess noise conditions. We assume that all users operate under the same excess noise $W$. It is observed that for the received power  \( P_{\max}^{(\mathrm{BS})} <-80~\text{dBm}\), the asymptotic Holevo information based on the large modulation approximation is lower than the explicit Holevo information. This is because quantum noise dominates the Holevo information in the low-power regimes. As the received power increases to \( P_{\max}^{(\mathrm{BS})} >-80~\text{dBm}\), the large modulation approximation concurs with the explicit Holevo information, validating the effectiveness of $\tilde{\mathcal{X}}\big( Y_k;E_k \big)$ in \eqref{eq:asympt_Holevo} using large modulation approximation.

\begin{figure}
\centering
\includegraphics[width=3.5in]{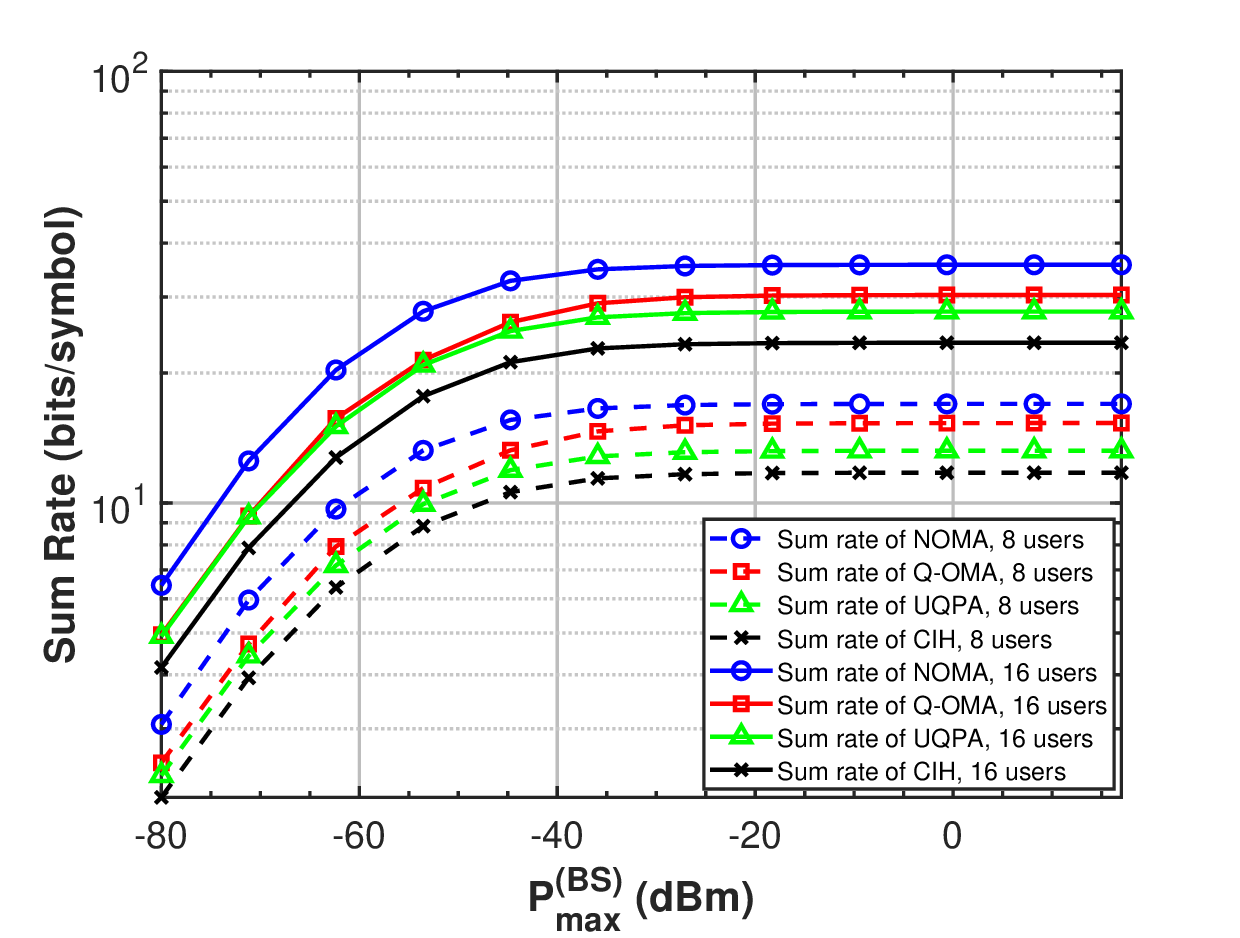}
\caption{Comparison of sum SKR between the proposed NONA-CVQKD (with Algorithm 1), Q-OMA, UQPA, and CIH schemes.}
\label{sum_rate}
\end{figure}

Fig.~\ref{sum_rate} evaluates the sum SKR of the proposed NONA-CVQKD scheme in comparison with the three benchmarks. The NOMA-CVQKD system, substantiated by Algorithm 1, exhibits the fastest sum SKR growth among all schemes and achieves increasingly significant gains as \( P_{\max}^{(\mathrm{BS})} \) increases. 
This is because the NOMA-CVQKD scheme utilizes power-domain multiplexing to enable simultaneous multi-user transmission, thereby achieving higher sum-rate in \eqref{eq:SumRate}, while employing SIC to effectively eliminate multi-user interference.

\section{Conclusion}
\label{sec:Conclusion}
This paper has proposed a novel uplink NOMA-CVQKD transmission framework designed to operate securely under collective attacks, atmospheric turbulence, and excess noise impairments. By employing Gaussian-modulated coherent states and an SIC-based heterodyne receiver, we derived rigorous asymptotic bounds on the legitimate users’ achievable SKRs and the eavesdropper’s accessible information. We further formulated a power allocation problem and developed an SCA-based algorithm that converges to a locally optimal solution satisfying the KKT conditions. Extensive simulations have demonstrated that the proposed scheme achieves up to a 23\% improvement in sum SKR compared to conventional quantum OMA methods, while supporting more users and maintaining robustness under varying channel conditions. Future work will focus on extending the framework to incorporate finite-key effects, adaptive modulation schemes, and experimental validations through both quantum channel emulation and a lab-scale prototype system, ultimately enhancing the practicality of NOMA-CVQKD systems for large-scale quantum networks.

\appendices
\section{Proof of Lemma 1}
 We employ the quantum mutual information to characterize the system's sum-rate \cite{Quantum_Communication,Cambridge}. In the SIC process, the achievable rate of the user with the highest channel gain is  
    \begin{equation}
    I_1(X_1; Y) = H(Y) - H(Y|X_1),
    \end{equation}
   where $H(\cdot )$ gives Shannon entropy, $H\left( \cdot |\cdot \right) $ stands for conditional entropy, and $I\left( \cdot \right) $ stands for mutual information. 

    When evaluating the \(k\)-th user,  the information of the first \((k-1)\) users is decoded and canceled. The conditional entropy \(H(Y|X_1, \cdots, X_{k-1})\) is used to evaluate the mutual information of the \(k\)-th user. The remaining signals from the other \((K-k)\) users are treated as interference, yielding the rate expression of the \(k\)-th user, as follows:
    \begin{equation}
    \begin{aligned}
    &I_k(X_k; Y|X_1, \cdots, X_{k-1}) \\
    &= H(Y|X_1, \cdots, X_{k-1}) - H(Y|X_1, \cdots, X_k).
    \end{aligned}
    \end{equation}

    We obtain the sum-rate for the uplink NOMA-CVQKD system, as given by 
    \begin{subequations}
    \label{eq:SumRate_H}
    \begin{flalign}
    I_{\mathrm{sum}}&=I_1\big( X_1;Y\big)+\sum\limits_{k=2}^{K}{I_k\big( X_k;Y|X_1,\cdots,X_{k-1}\big)}\label{eq:SumRate,0}\\
    &=\underbrace{H\big(Y\big)-H\big(Y|X_1\big)}_{I_1\big( X_1;Y\big)}+\underbrace{H\big(Y|X_1\big)-H\big(Y|X_1,X_2\big)}_{I_2\big( X_2;Y|X_1\big)}+\notag\\
    &\!\cdots\!+\!\underbrace{H\big(Y|X_1,\!\cdots\!,X_{K-1}\big)\!-\!H\big(Y|X_1,\!\cdots\!,X_K\big)}_{I_K\big( X_K;Y|X_1,\cdots,X_{K-1}\big)}\label{eq:SumRate,1}\\
    &=H\big(Y\big)-H\big(Y|X_1,\cdots,X_K\big), \label{eq:SumRate,2}
    \end{flalign}
    \end{subequations}
    where \eqref{eq:SumRate,0} is based on the sum mutual information for all users; 
    \eqref{eq:SumRate,1} is based on the definition of mutual information and the SIC process. Simplifying \eqref{eq:SumRate,1}
yields \eqref{eq:SumRate,2}.

The PDF of the received coherent states follows Gaussian statistics, as given in \eqref{eq:basePr_group}. We can employ Shannon entropy to measure the information of the coherent states in the form of their uncertainty \cite{Survey}. The 
entropy $H(Y)$ and the conditional entropy $H\left(Y|X_1,\cdots,X_K\right)$ in \eqref{eq:SumRate,2} can be written as
\begin{align}
&H(Y)= -\int_{-\infty}^{\infty} \int_{-\infty}^{\infty} \cdots \int_{-\infty}^{\infty} 
       \Pr\big( \boldsymbol{X} \big) \Pr\big( Y|\boldsymbol{X} \big) \,d{X_1}\cdots d{X_K} \nonumber\\
     & \times \log \Big[ \!\int_{-\!\infty}^{\infty}\!\!\!\!\!\! \cdots \!\!\int_{\!\!-\infty}^{\infty}
      \!\!\!\! \Pr\big( \boldsymbol{X} \big) \Pr\big( Y|\boldsymbol{X} \big) \,d{X_1}\cdots d{X_K} \Big] dy;\label{entropy}\\
&H\big( Y|X_1,\cdots,X_K \big) = -\int_{-\infty}^{\infty} \int_{-\infty}^{\infty} \cdots \int_{-\infty}^{\infty} 
       \mathrm{Pr}\big( \boldsymbol{X} \big) \mathrm{Pr}\big( Y|\boldsymbol{X} \big) \nonumber\\
     &\quad \quad\quad\quad\quad\quad\quad\quad\times \log \mathrm{Pr}\big( Y|\boldsymbol{X} \big) \,d{X_1}\cdots d{X_K}dy.\label{conditional_entropy}
\end{align}
Substituting \eqref{eq:basePr_group}, \eqref{entropy}, and \eqref{conditional_entropy} into \eqref{eq:SumRate,2} leads to \eqref{eq:SumRate}. 

Next, we analyze the convergence of \( I_{\text{sum}} \) in \eqref{eq:SumRate}, as \( \|\boldsymbol{X}\| \to \infty \) and \(|y|\to\infty\); \( I_{\text{sum}} \) involves integrals over the joint PDF, \( \mathrm{Pr}(\boldsymbol{X}) \), and the conditional PDF, \( \mathrm{Pr}(Y|\boldsymbol{X}) \).
The joint PDF, \(\mathrm{Pr}(\boldsymbol{X})\), of \eqref{eq:joint_PDF} is the product of independent Gaussian distributions, where \(X_k\) has finite variance \(\delta_k^2>0,\,\forall k\) and satisfies \(\int_{-\infty}^{\infty} \mathrm{Pr}(\boldsymbol{X}) d\boldsymbol{X} = 1\). As  \(|X_k| \to \infty\), \(\mathrm{Pr}(\boldsymbol{X})\to0\) decays exponentially, ensuring the convergence of \(\mathrm{Pr}(\boldsymbol{X})\).
The conditional PDF of \(\mathrm{Pr} (Y|\boldsymbol{X})\) in \eqref{eq:basePr_group} is a Gaussian distribution with finite and positive variance \(\delta^2\). Given \(\boldsymbol{X}\), \(\mathrm{Pr}(Y|\boldsymbol{X})\) decays exponentially as \(|y| \to \infty\), due to the \(\exp\left( -\frac{y^2}{2\delta^2} \right)\) term.  
Given \(y\), \(\exp\left( -\frac{(y - \sum_{k=1}^K \sqrt{T_k} X_k)^2}{2\delta^2} \right)\) ensures exponential decay as \(|X_k| \to \infty\), ensuring the convergence of \(\mathrm{Pr} (Y|\boldsymbol{X})\).
As a result, both \(\mathrm{Pr}(\boldsymbol{X})\) and \(\mathrm{Pr}(Y|\boldsymbol{X})\) exhibit exponential decay in their respective domains, guaranteeing the convergence of the integrals in \( I_{\text{sum}} \)  as \(\|\boldsymbol{X}\|\to\infty\) and \(|y|\to\infty\).

\section{Proof of Lemma 2}
According to EPI, for two independent random variables, $X$ and $Y$, yield $e^{2H\left( X+Y \right)}\geqslant e^{2H\left( X \right)}+e^{2H\left( Y \right)}$. Using EPI, we can obtain the lower bound of the achievable key rate for the $k$-th user, as given by
\begin{equation}
\begin{aligned}
\label{eq:lower}
I_k &= I\big( X_k;Y_k \big) \\
    &= H\big( Y_k \big) - H\big( Y_k|X_k \big) \\
    &= H\big( X_k + n_k \big) - H\big( n_k \big) \\
    &\geqslant \frac{1}{2}\log \Big( e^{2H(X_k)} + e^{2H(n_k)} \Big) - H\big( n_k \big) \\
    &= \frac{1}{2}\log \Big( 1 + \frac{e^{2H(X_k)}}{e^{2H(n_k)}} \Big),
\end{aligned}
\end{equation}
where 
\begin{align}
\label{eq:h_X_k}
H\big( X_k \big) \!\!&=\! -\!\!\!\int_{-\infty}^{\infty}\!\!\!\!\!\Pr( X_k ) \log \Pr( X_k )dx \!=\!\!\frac{1}{2}\log ( 2\pi e \delta_k^2 ),
\end{align}
since \( X_k \) is the quadrature measurement outcome of the Gaussian-modulated coherent state \( |\alpha_{p_k,q_k}^{(k)}\rangle \), with the PDF \( \mathrm{Pr}\left( X_k \right) ={\frac{1}{\sqrt{2{\pi }V_a^{(k)}}}\exp \left( -\frac{{X_k}^2}{2{}V_a^{(k)}} \right)}\). 

Next, we derive the upper bound of $H\left( n_k \right)$  in \eqref{eq:lower}. According to the maximum entropy principle~\cite{Maximum_Entropy}, we have
\begin{equation}
\begin{aligned}
H\big( n_k \big) \leqslant \frac{1}{2}\log \big( 2\pi e \delta_k^2 \big).
\end{aligned}
\end{equation}
As a result, the upper bound of $H\left( n_k \right)$ can be written as
\begin{equation}
\begin{aligned}
\label{eq:h_Nk}
H\big( n_k \big) \!\leqslant\! \frac{1}{2}\log \Big\{ 2\pi e \Big[ \sum_{i\!=\!k\!+\!1}^K\!\! V_{a}^{(i)} \!\!+\! \big(1\!-\!T_k\big) W_k \!+ \!\delta_{\det}^2 \Big] \Big\}.
\end{aligned}
\end{equation}
Considering a heterodyne receiver measures the coherent-state signal, we substitute \eqref{eq:h_X_k} and \eqref{eq:h_Nk} into \eqref{eq:lower} to derive the lower bound of the achievable key rate for user \(k\), leading to~\eqref{eq:I_k_lower}.

\section{Proof of Theorem 1}
Given the achievable key rate lower bound $I_{k}^{\left( \mathrm{low} \right)}$ for the $k$-th user in \eqref{eq:I_k_lower} and the upper bound on Eve's intercepted information $\mathcal{X} \left( Y_k;E_k \right) $  from the $k$-th user in \eqref{eq:upper}, we obtain the sum SKR in \eqref{eq:SKR}, with its convergence properties analyzed in the following.

Let \( V_{a}^{(k)} = c_k V \), $\forall k$, where \( c_k \) is a constant. When the transmit power is large, i.e., \( V_{a}^{(k)} \to \infty \), the lower bound of the achievable key rate at user $k$ in \eqref{eq:I_k_lower} converges: 
\begin{equation}  
\label{eq:I_k_app}
I_{k}^{(\mathrm{low})} \xrightarrow{V_{a}^{(k)} \to \infty} \frac{1}{2} \log \Bigg[ 1 + \frac{T_k c_k}{\sum_{i=k+1}^K c_i} \Bigg].  
\end{equation} 
Thus, \( I_{k}^{(\mathrm{low})} \)  converges to a constant as $V_{a}^{(k)} \to \infty$.  

Next, we analyze the convergence of the Holevo information in \eqref{eq:upper}. As \( V_{a}^{(k)} \to \infty \), the symplectic eigenvalues of \eqref{eq:eigenvalues} exhibit  asymptotic convergence, as follows:
  \begin{equation}  
  \label{eq:h1_app}
\lambda_1^{(k)}\! \xrightarrow{V_{a}^{(k)} \!\to  \!\infty} V_{a}^{(k)}\! \sqrt{1 \!\!- \!\!2T_k \!+\! \big( T_k \!\!+ \!\!\!\sum_{i=\!k+\!1}^K \!\!\frac{c_i}{c_k} \big)^2}\!=\! V_{a}^{(k)}  C_1,  
\end{equation}  
\begin{equation} 
\label{eq:h2_app}
\lambda_2^{(k)} \xrightarrow{V_{a}^{(k)}\! \to \infty} \frac{W_k}{\sqrt{1 \!-\! 2T_k \!+ \!\big( T_k \!+\!\! \sum_{i=k+1}^K \frac{c_i}{c_k} \big)^2}}\! =\! C_2,  
\end{equation}
The conditional symplectic eigenvalue of \eqref{eq:eigenvalue_E} converges:  
 \begin{equation}  
 \label{eq:h3_app}
\lambda_{\mathrm{het}}^{(k)} \!\xrightarrow{V_{a}^{(k)} \!\to \infty}\! V_{a}^{(k)} \!\Big( 1 \!-\! \frac{T_k}{T_k \!+ \!\sum_{i=k+1}^K \frac{c_i}{c_k}} \Big) \!= \!V_{a}^{(k)} C_3.  
\end{equation}
By substituting \eqref{eq:h1_app}, \eqref{eq:h2_app}, and \eqref{eq:h3_app} into \eqref{eq:upper}, the Holevo information can be written as
\begin{equation}  
\label{eq:Holevo_app}
\mathcal{X}(Y_k; E_k) \xrightarrow{V_{a}^{(k)} \to \infty} \log_2 \Big( \frac{C_1}{C_3} \Big) + h(C_2),  
\end{equation}
where $C_1$, $C_2$, and $C_3$ are constants.
As revealed in \eqref{eq:I_k_app} and \eqref{eq:Holevo_app}, the sum SKR \( I_{\mathrm{sum}}^{(\mathrm{sec})} \) in \eqref{eq:SKR} converges, as $V_{a}^{(k)} \to \infty$. 

 \section{Convexity of Entropy Functions in Holevo Information}
We analyze the convexity of $\tilde{h}( \tilde{\lambda} _{1}^{( k )} )$, $h( \tilde{\lambda} _{2}^{( k )} )$, and $-h( \tilde{\lambda} _{\mathrm{het}}^{( k )} )$ in the asymptotic  Holevo information $\tilde{\mathcal{X}} ( Y_k;E_k )$ in \eqref{eq:asympt_Holevo}.

For $\tilde{h}(\tilde{\lambda}_{1}^{(k)})$, by taking the second-order derivative with respect to $V_a^{(k)}$, we have
\begin{equation}
\frac{d^2 \tilde{h}(\tilde{\lambda}_{1}^{(k)})}{d (V_a^{(k)})^2} = \frac{b_k^2 - (V_a^{(k)})^2}{(\tilde{\lambda}_{1}^{(k)})^4},
\end{equation}
where $\tilde{h}(\tilde{\lambda}_{1}^{(k)})$ exhibits convexity for $V_a^{(k)} < b_k$, concavity for $V_a^{(k)} > b_k$, and a critical inflection at $V_a^{(k)} = b_k$.

We analyze the convexity of $h( x ) $ in \eqref{eq:f_entropy}. Taking the first-order derivative, we obtain 
\begin{equation}
\label{eq:app_h1}
    h'\big( x \big) = \frac{1}{2}\log \Big( \frac{x+1}{x-1} \Big) > 0, \quad x > 1.
\end{equation}
Taking the second-order derivative, we have 
\begin{equation}
\label{eq:app_h2}
    h''\big( x \big) = -\frac{1}{\ln(2)} \cdot \frac{1}{x^2 - 1} < 0.
\end{equation}
As a result, $h( x ) $ is a concave function.

For $h( \tilde{\lambda}_{2}^{( k )} ) $, $h''\big( x \big) < 0$ in \eqref{eq:app_h2} indicates that $h( x ) $ is concave; $h'(x) > 0$ in \eqref{eq:app_h1} confirms that $h(x)$ is monotonically increasing for $x > 1$. Furthermore, \(\tilde{\lambda}_2^{(k)}(V_a^{(k)})\) is a concave function of \(V_a^{(k)}\) in \eqref{eq:eigenvalues_app2}, since  the second-order derivative satisfies $ \frac{d^2\tilde{\lambda}_{2}^{( k )}}{{dV_{a}^{( k )}}^2}<0$. 
According to the properties of composite functions,  $\tilde{h}( \tilde{\lambda}_{2}^{( k )} ) $ is a concave function in the feasible region.

We proceed to analyze the convexity of $-\tilde{h}( \tilde{\lambda}_{\mathrm{het}}^{( k )} ) $. 
By taking the first-order derivative of  \(\tilde{\lambda} _{\mathrm{het}}^{( k )}(V_a^{(k)})\) in \eqref{eq:eigenvalues_app3} with respect to $V_{a}^{( k )}$, we have 
\begin{equation}
    \frac{d\tilde{\lambda}_{\mathrm{het}}^{(k)}}{dV_{a}^{(k)}} = 1 - \frac{2T_k V_{a}^{(k)}}{b_k} + \frac{T_k^2 V_{a}^{(k)2}}{b_k^2}.
\end{equation}
   Taking the second-order derivative of \(\tilde{\lambda} _{\mathrm{het}}^{( k )}(V_a^{(k)})\) with respect to $V_{a}^{( k )}$ for $b_k>0$ and $ 0<T_k<1$, we have
\begin{equation}
\label{eq:app_lambda2}
    \frac{d^2\lambda_{\mathrm{het}}^{(k)}}{dV_{a}^{(k)2}} = -\frac{2T_k(1-T_k)(W + \delta_{\det}^2 + V_{\mathrm{I}})}{b_k^3} < 0. 
\end{equation}
By taking the second-order derivative of the composite function $-\tilde{h}( \tilde{\lambda}_{\mathrm{het}}^{( k )} ) $ with respect to $V_{a}^{( k )}$, it follows that  
\begin{equation}
     \frac{d^2}{dV_{a}^{(k)2}}\! \Big(\!-\!\tilde{h}\big( \tilde{\lambda}_{\mathrm{het}}^{(k)} \big)\! \Big) \!\!= \!
     -\tilde{h}''( \tilde{\lambda}_{\mathrm{het}}^{(k)} ) \Big( \frac{d\lambda_{\mathrm{het}}^{(k)}}{dV_{a}^{(k)}} \!\Big)^2 \!\!\!
     -\!\tilde{h}'\big( \tilde{\lambda}_{\mathrm{het}}^{(k)} \big) \frac{d^2\lambda_{\mathrm{het}}^{(k)}}{dV_{a}^{(k)2}},
\end{equation}
where $\tilde{h}^{\prime\prime}( \tilde{\lambda}_{\mathrm{het}}^{( k )} ) <0$, as given in \eqref{eq:app_h2}; $\tilde{h}^{\prime}( \tilde{\lambda}_{\mathrm{het}}^{( k )} )>0$, as given in \eqref{eq:app_h1}; and $\frac{d^2\lambda _{\mathrm{het}}^{( k )}}{{dV_{a}^{( k )}}^2}<0$, as given in \eqref{eq:app_lambda2}.
As a result, $\frac{d^2}{{dV_{a}^{( k )}}^2}( -\tilde{h}( \tilde{\lambda}_{\mathrm{het}}^{( k )} ) ) >0$. Thus, $-\tilde{h}( \tilde{\lambda}_{\mathrm{het}}^{( k )} ) $ is a convex function within its feasible domain.

\bibliographystyle{IEEEtran}
\bibliography{myreferences}
\end{document}